\documentclass[11pt,twoside,a4paper,notitlepage]{article}

\def\DONOTINSERTCOMMENTS{}

\usepackage{ifthen}

\newboolean{conferenceversion}
\setboolean{conferenceversion}{false}

\ifthenelse{\boolean{conferenceversion}}
{}
{
  \usepackage[a4paper,margin=2.5cm]{geometry}
\usepackage{tgtermes}
  \usepackage[scale=.92]{tgheros}
  \usepackage{tgcursor}
}

\usepackage{microtype}

\ifthenelse{\boolean{conferenceversion}}
{}
{
}

\ifthenelse{\boolean{conferenceversion}}
{}
{
  \usepackage[english]{babel}
}

\ifthenelse{\boolean{conferenceversion}}
{}
{
  \usepackage{fancyhdr}
}

\usepackage{subcaption}

\usepackage{bussproofs}

\ifthenelse{\boolean{conferenceversion}}
{}
{
  
}

\usepackage{ifthen}

\usepackage{ifthen}

\newboolean{detectedSTOC}
\newboolean{detectedFOCS}
\newboolean{detectedElsevier}
\newboolean{detectedNOW}
\newboolean{detectedLMCS}
\newboolean{detectedIEEE}
\newboolean{detectedPoster}
\newboolean{detectedSIAM}
\newboolean{detectedLNCS}
\newboolean{detectedACM}
\newboolean{detectedACMconf}
\newboolean{detectedSigplanconf}
\newboolean{detectedToC}
\newboolean{detectedLIPIcs}
\newboolean{detectedAAAI}
\newboolean{detectedIJCAI}
\newboolean{detectedCompCplx}
\newboolean{detectedEasyChair}
\newboolean{detectedArticle}
\newboolean{detectedReport}
\newboolean{detectedThesis}

\makeatletter

\@ifclassloaded{sig-alternate}
{\setboolean{detectedSTOC}{true}}
{\setboolean{detectedSTOC}{false}}

\@ifclassloaded{elsarticle}
{\setboolean{detectedElsevier}{true}}
{\setboolean{detectedElsevier}{false}}

\@ifclassloaded{now}
{\setboolean{detectedNOW}{true}}
{\setboolean{detectedNOW}{false}}

\@ifclassloaded{lmcs}
{\setboolean{detectedLMCS}{true}}
{\setboolean{detectedLMCS}{false}}

\@ifclassloaded{IEEEtran} {
  \setboolean{detectedIEEE}{true}
  \ifCLASSOPTIONconference {
    \setboolean{detectedFOCS}{true}
  }
  \else {
    \setboolean{detectedFOCS}{false}
  }
  \fi
}
{
  \setboolean{detectedFOCS}{false}
  \setboolean{detectedIEEE}{false}
}

\@ifclassloaded{siamart171218}
{\setboolean{detectedSIAM}{true}}
{\setboolean{detectedSIAM}{false}}

\@ifclassloaded{llncs}
{\setboolean{detectedLNCS}{true}}
{\setboolean{detectedLNCS}{false}}

\@ifclassloaded{acmsmall}
{\setboolean{detectedACM}{true}}
{\setboolean{detectedACM}{false}}

\@ifclassloaded{acmart}
{\setboolean{detectedACMconf}{true}}
{\setboolean{detectedACMconf}{false}}

\@ifclassloaded{sigplanconf}
{\setboolean{detectedSigplanconf}{true}}
{\setboolean{detectedSigplanconf}{false}}

\@ifclassloaded{toc}
{\setboolean{detectedToC}{true}}
{\setboolean{detectedToC}{false}}

\@ifclassloaded{lipics}
{\setboolean{detectedLIPIcs}{true}}
{\@ifclassloaded{lipics-v2016}
  {\setboolean{detectedLIPIcs}{true}}
  {\@ifclassloaded{oasics-v2016}
    {\setboolean{detectedLIPIcs}{true}}
    {\setboolean{detectedLIPIcs}{false}}
  }
}

\@ifclassloaded{cc}
{\setboolean{detectedCompCplx}{true}}
{\setboolean{detectedCompCplx}{false}}

\@ifclassloaded{easychair}
{\setboolean{detectedEasyChair}{true}}
{\setboolean{detectedEasyChair}{false}}

\@ifpackageloaded{aaai}
{\setboolean{detectedAAAI}{true}}        
{\@ifpackageloaded{aaai18}
  {\setboolean{detectedAAAI}{true}}
  {\setboolean{detectedAAAI}{false}}}

\@ifpackageloaded{ijcai18}
{\setboolean{detectedIJCAI}{true}}        
{\setboolean{detectedIJCAI}{false}}

\@ifclassloaded{sciposter}
{\setboolean{detectedPoster}{true}}
{\setboolean{detectedPoster}{false}}

\@ifclassloaded{article}
{\setboolean{detectedArticle}{true}}
{\setboolean{detectedArticle}{false}}

\makeatother

\ifthenelse{\not \isundefined{\examen} 
  \and \not \isundefined{\disputationsdatum} 
  \and \not \isundefined{\disputationslokal}}   
  {\setboolean{detectedThesis}{true}}
  {\setboolean{detectedThesis}{false}}

\ifthenelse{\boolean{detectedArticle} \or \boolean{detectedThesis}
  \or \boolean{detectedSTOC} \or \boolean{detectedFOCS}
  \or \boolean{detectedSIAM} \or \boolean{detectedIEEE}
  \or \boolean{detectedPoster}}
{\setboolean{detectedReport}{false}}
{\setboolean{detectedReport}{true}}

\ifthenelse{\boolean{detectedLIPIcs}}
{}
{\usepackage[utf8]{inputenc}}

\ifthenelse{\boolean{detectedIEEE}}
{\usepackage[cmex10]{amsmath}
  \interdisplaylinepenalty=2500}
{\usepackage{amsmath}}
\usepackage{amssymb}
\usepackage{amsfonts}

\usepackage{mathtools}

\usepackage{bussproofs}

\usepackage{bm}  

\usepackage{subcaption}

\ifthenelse
{     \boolean{detectedElsevier} \or \boolean{detectedSIAM}
  \or \boolean{detectedSIAM}} 
{}
{\usepackage{varioref}}

\usepackage{xspace}

\ifthenelse    
{\boolean{detectedSTOC}      \or \boolean{detectedSIAM}         \or 
  \boolean{detectedLMCS}     \or \boolean{detectedNOW}          \or 
  \boolean{detectedLNCS}     \or \boolean{detectedACM}          \or
  \boolean{detectedToC}      \or \boolean{detectedLIPIcs}       \or
  \boolean{detectedACMconf}  \or 
  \boolean{detectedCompCplx} \or \boolean{detectedSigplanconf}}
{}
{
  \usepackage{amsthm}
  \usepackage{amsthm-boldfont}
}

\ifthenelse        
{\boolean{detectedElsevier}  \or \boolean{detectedFOCS}         \or 
  \boolean{detectedSTOC}     \or \boolean{detectedPoster}       \or
  \boolean{detectedSIAM}     \or \boolean{detectedLMCS}         \or
  \boolean{detectedIEEE}     \or \boolean{detectedNOW}          \or
  \boolean{detectedToC}      \or \boolean{detectedThesis}       \or
  \boolean{detectedLNCS}     \or \boolean{detectedACM}          \or 
  \boolean{detectedLIPIcs}   \or \boolean{detectedAAAI}         \or
  \boolean{detectedACMconf}  \or \boolean{detectedIJCAI}        \or 
  \boolean{detectedCompCplx} \or \boolean{detectedSigplanconf}}
{}
{\usepackage{jncaptionsheadings}}

\usepackage{ifthen}
\usepackage{xspace}
\ifthenelse
{\boolean{detectedElsevier} \or \boolean{detectedSIAM}}
{}
{\usepackage{varioref}}

\DeclareMathAlphabet{\mathsfsl}{OT1}{cmss}{m}{sl}

\newcommand{\formatfunctiontoset}[1]{\mathit{#1}}

\newcommand{\formuladots}{\cdots}

\newcommand{\BIGOH}[1]{\mathrm{O} \left( #1 \right)}
\newcommand{\Bigoh}[1]{\mathrm{O} \bigl( #1 \bigr)}
\newcommand{\bigoh}[1]{\mathrm{O} ( #1 )}

\newcommand{\littleoh}[1]{\mathrm{o} ( #1 )}

\newcommand{\Bigtheta}[1]{\Theta \bigl( #1 \bigr)}
\newcommand{\bigtheta}[1]{\Theta ( #1 )}

\newcommand{\Bigomega}[1]{\Omega \bigl( #1 \bigr)}
\newcommand{\bigomega}[1]{\Omega ( #1 )}

\newcommand{\littleomega}[1]{\omega ( #1 )}

\ifthenelse{\boolean{detectedToC}}{}
{

  \newcommand{\Nplus}     {\mathbb{N}^{+}}

}

\newcommand{\ceiling}[1]{\lceil #1 \rceil}

\newcommand{\floor}[1]{\lfloor #1 \rfloor}

\newcommand{\MAXOFEXPR}[2][]{\max_{#1} \left\{ #2 \right\}}
\newcommand{\MINOFEXPR}[2][]{\min_{#1} \left\{ #2 \right\}}
\newcommand{\Maxofexpr}[2][]{\max_{#1} \bigl\{ #2 \bigr\}}
\newcommand{\Minofexpr}[2][]{\min_{#1} \bigl\{ #2 \bigr\}}
\newcommand{\maxofexpr}[2][]{\max_{#1} \{ #2 \}}

\newcommand{\MAXOFSET}[3][:]{\ifthenelse{\equal{#1}{;}}{\MAXOFEXPR{ #2 \,;\, #3 }}
     {\ifthenelse{\equal{#1}{:}}{\MAXOFEXPR{ #2 \,:\, #3 }}
     {\max \twincommandJN{\left\{}{#2}{\left#1}{\right}{\,#3}{\right\}}}}}
\newcommand{\MINOFSET}[3][:]{\ifthenelse{\equal{#1}{;}}{\MINOFEXPR{ #2 \,;\, #3 }}
     {\ifthenelse{\equal{#1}{:}}{\MINOFEXPR{ #2 \,:\, #3 }}
     {\min \twincommandJN{\left\{}{#2}{\left#1}{\right}{\,#3}{\right\}}}}}

\newcommand{\Maxofset}[3][:]{\ifthenelse{\equal{#1}{;}}{\Maxofexpr{ #2 \,;\, #3 }}
     {\ifthenelse{\equal{#1}{:}}{\Maxofexpr{ #2 \,:\, #3 }}
     {\max \twincommandJN{\bigl\{}{#2}{\bigl#1}{\bigr}{\,#3}{\bigr\}}}}}
\newcommand{\Minofset}[3][:]{\ifthenelse{\equal{#1}{;}}{\Minofexpr{ #2 \,;\, #3 }}
     {\ifthenelse{\equal{#1}{:}}{\Minofexpr{ #2 \,:\, #3 }}
     {\min \twincommandJN{\bigl\{}{#2}{\bigl#1}{\bigr}{\,#3}{\bigr\}}}}}

\newcommand{\F}{\mathbb{F}}

\ifthenelse{\boolean{detectedLMCS}}
{}
{}

\newcommand{\twincommandJN}[6]{#1#2#3\vphantom{#2#5}\mspace{-2.05mu}#4.#5#6}

\newcommand{\edges}[1]{E( #1 )}

\newcommand{\vertices}[1]{V( #1 )}

\newcommand{\set}[1]{\{ #1 \}}
\newcommand{\Set}[1]{\bigl\{ #1 \bigr\}}

\newcommand{\setdescr}[3][\mid]{\set{ #2 #1 #3 }}
\newcommand{\Setdescr}[3][|]{\ifthenelse{\equal{#1}{;}}{\Set{ #2 \,;\, #3 }}
     {\ifthenelse{\equal{#1}{:}}{\Set{ #2 \,:\, #3 }}
     {\twincommandJN{\bigl\{}{#2\,}{\bigl#1}{\bigr}{\,#3}{\bigr\}}}}}

\newcommand{\Setsize}[1]{\bigl\lvert#1\bigr\rvert}
\newcommand{\setsize}[1]{\lvert#1\rvert}

\newcommand{\union}{\cup}
\newcommand{\Union}{\bigcup}

\newcommand{\unionSP}{\, \union \, }

\newcommand{\Lor}{\bigvee}
\newcommand{\Land}{\bigwedge}

\newcommand{\olnot}[1]{\overline{#1}}

\newcommand{\clwidth}{k}

\newcommand{\xcnfform}[1]{\mbox{\ensuremath{#1}-CNF} formula\xspace}
\newcommand{\kcnfform}{\xcnfform{\clwidth}}
\newcommand{\xclause}[1]{\mbox{\ensuremath{#1}-clause}\xspace}

\newcommand{\complclassformat}[1]{\textrm{\upshape{\textsf{#1}}}\xspace}

\newcommand{\PSPACE}{\complclassformat{PSPACE}}

\newcommand{\introduceterm}[1]{{\emph{#1}}}

\newcommand{\eqperiod}{\enspace .}
\newcommand{\eqcomma}{\enspace ,}

\ifthenelse{\boolean{detectedToC}}{}
  { }
\ifthenelse{\boolean{detectedToC}}{}{\newcommand{\ie}{i.e.,\ }}

\ifthenelse{\boolean{detectedLIPIcs} \or \boolean{detectedIJCAI}}
{}
{}

\newcommand{\etal}{et al.\@\xspace}

\ifthenelse{\boolean{detectedIEEE}}{}{}

\newcommand{\refsec}[1]{Section~\ref{#1}}

\newcommand{\Refsec}[1]{Section~\ref{#1}}

\newcommand{\reffig}[1]{Figure~\ref{#1}}

\ifthenelse
{\isundefined{\refeq}}
{\newcommand{\refeq}[1]{\eqref{#1}}}
{\renewcommand{\refeq}[1]{\eqref{#1}}}

\ifthenelse{\isundefined{\DONOTLOADHYPERREFPACKAGE}
  \and \not \boolean{detectedSTOC}        \and \not \boolean{detectedFOCS}
  \and \not \boolean{detectedPoster}      \and \not \boolean{detectedElsevier} 
  \and \not \boolean{detectedSIAM}        \and \not \boolean{detectedACM}
  \and \not \boolean{detectedIEEE}        \and \not \boolean{detectedNOW}
  \and \not \boolean{detectedToC}         \and \not \boolean{detectedThesis}
\and \not \boolean{detectedLIPIcs}      \and \not \boolean{detectedSIAM}
  \and \not \boolean{detectedAAAI}        \and \not \boolean{detectedIJCAI}
  \and \not \boolean{detectedSigplanconf} \and \not \boolean{detectedACMconf}   
  \and \not \boolean{detectedCompCplx} \and \not \boolean{detectedEasyChair}}
{\usepackage[colorlinks=true,citecolor=blue]{hyperref}}
{}

\ifthenelse{\boolean{detectedSTOC}}                  
{\input{definitions-STOC.tex}}
{\ifthenelse{\boolean{detectedSIAM}}
  {\input{definitions-SIAM.tex}}
  {\ifthenelse{\boolean{detectedNOW}}
    {\input{definitions-NOW.tex}}
    {\ifthenelse{\boolean{detectedLMCS}}
      {\input{definitions-lmcs.tex}}
      {\ifthenelse{\boolean{detectedIEEE}}
        {\input{definitions-IEEE.tex}}
        {\ifthenelse{\boolean{detectedLNCS}}
          {\input{definitions-LNCS.tex}}
          {\ifthenelse{\boolean{detectedACM} \or \boolean{detectedACMconf}}
            {\input{definitions-ACM.tex}}
            {\ifthenelse{\boolean{detectedToC}}
              {\input{definitions-ToC.tex}}
              {\ifthenelse{\boolean{detectedLIPIcs}}
                {

\theoremstyle{plain}    

\newtheorem{fact}[theorem]{Fact}
\newtheorem{observation}[theorem]{Observation}
 }
                {\ifthenelse{\boolean{detectedCompCplx}}
                  {\input{definitions-CompCplx.tex}}
                  {\ifthenelse{\boolean{detectedSigplanconf}}
                    {\input{definitions-Sigplanconf.tex}}
                    {\ifthenelse{\boolean{detectedArticle} 
                        \or \boolean{detectedElsevier}
                        \or \boolean{detectedEasyChair}}
                      {

\newtheorem{standardlocalcounter}{Dummy}[section]

\theoremstyle{plain}    

\newtheorem{theorem}[standardlocalcounter]{Theorem}
\newtheorem{lemma}[standardlocalcounter]{Lemma}
\newtheorem{proposition}[standardlocalcounter]{Proposition}
\newtheorem{corollary}[standardlocalcounter]{Corollary}

\theoremstyle{definition}

\newtheorem{definition}[standardlocalcounter]{Definition}
\newtheorem{claim}[standardlocalcounter]{Claim}

\theoremstyle{remark}
\newtheorem{remark}[standardlocalcounter]{Remark}

 }
                      {\input{definitions-report.tex}}}}}}}}}}}}}

\ifthenelse{\boolean{detectedThesis}}
{}
{\usepackage{ifpdf}}

\ifthenelse{\boolean{detectedToC}}
{}
{\usepackage{graphicx}}

\ifpdf         
\fi

\ifthenelse
{\boolean{detectedSTOC}   \or \boolean{detectedThesis} \or 
 \boolean{detectedLNCS}   \or \boolean{detectedToC}    \or 
 \boolean{detectedLIPIcs} \or \boolean{detectedAAAI}   \or
 \boolean{detectedIJCAI}  \or \boolean{detectedSIAM}}
{}
{
  \setcounter{secnumdepth}{3}
  \setcounter{tocdepth}{3}
}

\newcount\hours
\newcount\minutes
\def\SetTime{\hours=\time
\global\divide\hours by 60
\minutes=\hours
\multiply\minutes by 60
\advance\minutes by-\time
\global\multiply\minutes by-1 }
\SetTime
\def\now{\number\hours:\ifnum\minutes<10 0\fi\number\minutes}

\newcommand{\pyramidgraph}[1][]{\Pi_{#1}}

\newcommand{\formf}{\ensuremath{F}}

\newcommand{\varx}{\ensuremath{x}}
\ifthenelse{\boolean{detectedACMconf}}
{\providecommand{\vary}{\ensuremath{y}}}
{\newcommand{\vary}{\ensuremath{y}}}

\ifthenelse{\boolean{detectedACMconf}}
{
 }
{
 }

\newcommand{\lita}{\ensuremath{a}}

\newcommand{\clc}{\ensuremath{C}}
\newcommand{\cld}{\ensuremath{D}}

\newcommand{\spacestd}{s}

\newcommand{\formulaformat}[1]{\mathit{#1}}

\newcommand{\pebcontr}[2][G]{\ensuremath{\formulaformat{Peb}^{#2}_{#1}}}

\providecommand{\partassign}{\rho}

\newcommand{\monsize}[1]{\mathit{mSize}(#1)}
\newcommand{\Monsize}[1]{\mathit{mSize}\bigl(#1\bigr)}

\providecommand{\varvecx}{\vec{x}}

\providecommand{\polyp}{p}
\providecommand{\polyq}{q}
\providecommand{\polyr}{r}
\providecommand{\polys}{s}

\providecommand{\polyset}{\mathcal{P}}

\providecommand{\pdegree}[1]{\mathrm{deg}(#1)}

\providecommand{\degd}{d}
\providecommand{\designd}{D}

\newcommand{\visitprice}[2][]{\formatpebblingprice{RPeb}^V_{#1}\ifthenelse{\equal{#1}{}}{\!}{}(#2)}

\newcommand{\prednode}[2][G]{\formatfunctiontoset{pred}_{#1}(#2)}

\newcommand{\formatpebblingstrategy}[1]{\mathcal{#1}}

\newcommand{\pebbling}[1][P]{\formatpebblingstrategy{#1}}
\providecommand{\formatconfiguration}[1]{\mathbb{#1}}
\newcommand{\pconf}[1][P]{\formatconfiguration{#1}}
\newcommand{\pconfafter}[1][P]{\formatconfiguration{#1}_{i+1}}
\newcommand{\pconfbefore}[1][P]{\formatconfiguration{#1}_i}

\newcommand{\formatpebblingprice}[1]{\textsl{\textsf{#1}}}

\newcommand{\stoptime}{\tau}
\renewcommand{\stoptime}{t}

\newcommand{\pebspace}[1]{\formatpebblingprice{space} ( #1 )}

\newcommand{\pebtime}[1]{\formatpebblingprice{time} ( #1 )}

\newcommand{\bP}{\mathbb{P}}
\newcommand{\pred}{\mathrm{pred}}
\newcommand{\sink}{\mathrm{sink}}\newcommand{\bC}{\mathbb{C}}
\newcommand{\bF}{\mathbb{F}}
\newcommand{\bH}{\mathbb{H}}

\newcommand{\sourcestd}{s}
\newcommand{\sinkstd}{z}
\newcommand{\sourcesetstd}{S}
\newcommand{\sinksetstd}{Z}

\newcommand{\sizerelatedparam}{m}

\newcommand{\levelparam}{r}

\newcommand{\indegreedag}{\ell}

\newcommand{\depthdag}{d}

\newcommand{\graphcopyindex}[1]{(#1)}

\newcommand{\singlesinkdagtext}{single-sink DAG\xspace}

\newcommand{\nsupconctext}[1]{\ensuremath{#1}-super\-concen\-trator\xspace}

\newcommand{\csnspines}{c}   \newcommand{\csrec}{\levelparam}       \newcommand{\csdag}[2]{\Gamma({#1},{#2})}
\newcommand{\cssink}[1][]{\gamma_{#1}}

\newcommand{\cspebsurplus}{\spacestd}
\newcommand{\csfunc}{g}   

\newcommand{\ltsupconcsize}{\sizerelatedparam}
\newcommand{\ltncopies}{\levelparam}
\newcommand{\ltsscdag}[2]{\Phi({#1},{#2})}

\newcommand{\suchthat}{such that\xspace}

\usepackage[usenames,dvipsnames]{color}
\usepackage{numname}
\newcounter{authorcount}
\newcommand{\newauthor}[3]{\newcounter{#1comment}
\expandafter\newcommand\csname #1comment\endcsname[1]{\ifdefined\DONOTINSERTCOMMENTS\relax\else \medskip\par\noindent
{\bfseries \scshape \footnotesize #2's comment
\stepcounter{#1comment}\csname the#1comment\endcsname}:
{\sffamily \itshape \scriptsize\textcolor{#3}{##1}\par}
\medskip \fi}
\stepcounter{authorcount}
\expandafter\edef\csname #1ordinal\endcsname{\theauthorcount}
\expandafter\newcommand\csname theauthor#1\endcsname{the \ordinaltoname{\csname #1ordinal\endcsname} author\xspace}
\expandafter\newcommand\csname Theauthor#1\endcsname{The \ordinaltoname{\csname #1ordinal\endcsname} author\xspace}
}

\newauthor{sfdr}{Susanna}{OliveGreen}
\newauthor{jn}{Jakob}{blue}
\newauthor{om}{Or}{Yellow}
\newauthor{rr}{Robert}{Orchid}

\ifthenelse{\boolean{conferenceversion}}
{

}
{

}

\ifthenelse{\boolean{conferenceversion}}
{}
{
  \numberwithin{equation}{section}
}

\usepackage{tikz}

\usepackage{xparse,expl3}

\ExplSyntaxOn
\NewDocumentCommand{\reverse}{m}
{
  \tl_reverse:N #1
  \tl_use:N #1
}
\NewDocumentCommand{\reverseto}{mm}
{
  \tl_reverse:N #1
  \tl_set_eq:NN #2 #1
}
\ExplSyntaxOff

\begin{document}

\title{Nullstellensatz Size-Degree Trade-offs from Reversible Pebbling\footnote{A preliminary version of this work appeared in CCC 2019.}}

\author{Susanna F. de Rezende \\
Mathematical Institute of the \\ Czech Academy of Sciences
  \and
  Or Meir\\
  University of Haifa
  \and
  Jakob Nordström \\
  University of Copenhagen and \\
  KTH Royal Institute of Technology
  \and
  Robert Robere\\
  DIMACS and\\
Institute for Advanced Study}

\date{\today}

\maketitle

\ifthenelse{\boolean{conferenceversion}}
{}
{
\thispagestyle{empty}

\pagestyle{fancy}
\fancyhead{}
  \fancyfoot{}
\renewcommand{\headrulewidth}{0pt}
  \renewcommand{\footrulewidth}{0pt}
 
\fancyhead[CE]{\slshape 
    NULLSTELLENSATZ SIZE-DEGREE TRADE-OFFS FROM REVERSIBLE PEBBLING
  }
  \fancyhead[CO]{\slshape \nouppercase{\leftmark}}
  \fancyfoot[C]{\thepage}
  
\setlength{\headheight}{13.6pt}
}

\begin{abstract}
  We establish an exactly tight relation between reversible pebblings
  of graphs and Nullstellensatz refutations of pebbling formulas,
  showing that a graph~$G$ can be reversibly pebbled in time~$t$ and
  space~$s$ if and only if there is a Nullstellensatz refutation of the
  pebbling formula over~$G$ in size $t+1$ and degree~$s$
  (independently of the field in which the Nullstellensatz refutation
  is made). We use this correspondence to prove a number of strong
  size-degree trade-offs for Nullstellensatz, which to the best of our
  knowledge are the first such results for this proof system.
\end{abstract}

\section{Introduction}
\label{sec:intro}

In this work, we obtain strong trade-offs in proof complexity by
making a connection to pebble games played on graphs. In this
introductory section we start with a brief overview of these two areas
and then explain how our results follow from connecting the two.

\subsection{Proof Complexity}

Proof complexity is the study of efficiently verifiable certificates
for mathematical statements.  More concretely, statements of interest
claim to provide correct answers to questions like:
\begin{itemize}
\item
  Given a CNF formula, does it have a satisfying assignment or not?
\item
  Given a set of polynomials over some finite field, do they have a common root?
\end{itemize}
There is a clear asymmetry here in that it seems obvious what an
easily verifiable certificate for positive answers to the above
questions should be, while it is not so easy to see what a concise
certificate for a negative answer could look like. The focus of proof
complexity is therefore on the latter scenario.

In this paper we study the algebraic proof system
system \introduceterm{Nullstellensatz} introduced by Beame
\etal~\cite{BIKPP94LowerBounds}.
A
\introduceterm{Nullstellensatz refutation}
of a set of polynomials
$\polyset = \setdescr{\polyp_i}{i \in [m]}$
with coefficients in a field~$\F$
is an expression
\begin{equation}
  \label{eq:ns-refutation}
  \sum_{i=1}^{m} \polyr_i \cdot \polyp_i
  +
  \sum_{j=1}^{n} \polys_j  \cdot (\varx_j^2 - \varx_j)
  = 1
\end{equation}
(where
$\polyr_i,\polys_j$
are also polynomials),
showing that
$1$ lies in the polynomial ideal in the
ring~$\F[\varx_1, \ldots,\varx_n]$  generated by
$
\polyset \union
\Setdescr{\varx^2_j - \varx_j}{j \in [n]}
$.
By (a slight extension of) Hilbert's Nullstellensatz, such a
refutation exists if and only if there is no common $\set{0,1}$-valued
root for the  set of polynomials~$\polyset$.

Nullstellensatz can also be viewed as a proof system for certifying
the unsatisfiability of CNF formulas. If we translate a clause like,
e.g.,
$\clc = x \lor y \lor \olnot{z}$
to the polynomial
$
\polyp({\clc})
=
(1-x)(1-y)z
=
z - yz -xz + xyz
$, then an assignment to the variables in a CNF formula
$\formf = \Land_{i=1}^{m} \clc_i$
(where we think of $1$ as true and $0$ as false)
is satisfying precisely if all the polynomials
$\setdescr{\polyp({\clc_i})}{i \in [m]}$
vanish.

The \introduceterm{size} of a Nullstellensatz refutation
\refeq{eq:ns-refutation}
is the total number of monomials in all the polynomials
$\polyr_i \cdot \polyp_i$
and
$\polys_j  \cdot (\varx_j^2 - \varx_j)$
expanded out as linear combinations of monomials.
Another, more well-studied, complexity measure for Nullstellensatz is
\introduceterm{degree},
which
is defined as
$
\maxofexpr{
  \pdegree{\polyr_i \cdot \polyp_i},
  \pdegree{\polys_j \cdot (\varx_j^2 - \varx_j)}
}
$.

In order to prove a lower bound~$d$ on the Nullstellensatz degree of
refuting a set of polynomials~$\polyset$, one can construct a
\introduceterm{$\degd$-design},
which is a map~$\designd$ from degree-$\degd$ polynomials in
$\F[\varx_1, \ldots,\varx_n]$
to~$\F$ such that
\begin{enumerate}
\item
  $\designd$ is linear, \ie
  $
  \designd(\alpha \polyp + \beta \polyq) =
  \alpha \designd(\polyp) + \beta \designd(\polyq)
  $
  for $\alpha,\beta \in \F$;

\item
  $\designd(1) = 1$;

\item
  $\designd(\polyr \polyp) = 0$
  for all
  $\polyp \in \polyset$
  and
  $\polyr \in
  \F[\varx_1, \ldots,\varx_n]
$
  such that
  $\pdegree{\polyr \polyp} \leq \degd$;

\item
  $\designd(\varx^2 \polys) = \designd(\varx \polys )$
  for all
  $\polys \in
  \F[\varx_1, \ldots,\varx_n]
$
  such that
  $\pdegree{\polys} \leq \degd - 2$.
\end{enumerate}
Designs provide a characterization of Nullstellensatz degree in that
there is a \mbox{$\degd$-design} for $\polyset$ if and
only if there is no Nullstellensatz refutation of~$\polyset$ in
degree~$\degd$~\cite{Buss98LowerBoundsNS}.
Another possible approach to prove degree lower bounds is by
computationally efficient versions of Craig's interpolation
theorem. It was shown in~\cite{PS98Algebraic} that constant-degree
Nullstellensatz refutations yield polynomial-size monotone span
programs, and that this is also tight: every span program is a unique
interpolant for some set of polynomials refutable by Nullstellensatz.
This connection has not been used to obtain Nullstellensatz degree
lower bounds, however, due to the difficulty of proving span program
lower bounds.

Lower bounds on Nullstellensatz degree have been proven for sets of
polynomials encoding combinatorial principles such as the pigeonhole
principle~\cite{BCEIP98Relative},
induction principle~\cite{BP98Good},
house-sitting principle~\cite{CEI96Groebner,Buss98LowerBoundsNS},
matching~\cite{BIKPRS97ProofCplx},
and pebbling~\cite{BCIP02Homogenization}.
It seems fair to say that research in algebraic proof complexity
soon  moved on to stronger systems such as
\introduceterm{polynomial calculus}~\cite{CEI96Groebner,ABRW02SpaceComplexity},
where the proof that $1$ lies in the ideal generated by
$
\polyset \union
\Setdescr{\varx^2_j - \varx_j}{j \in [n]}
$
can be constructed dynamically by a step-by-step derivation.
However, the Nullstellensatz proof system has been the focus of
renewed interest in a recent line of works
\cite{RPRC16Exponential,PR17StronglyExponential,PR18LiftingNS,dRMNPRV19Lifting}
showing that Nullstellensatz lower bounds can be lifted to
stronger lower bounds for more powerful computational models
using composition with gadgets.
The size complexity measure for Nullstellensatz has also received
attention in recent papers such as
\cite{Berkholz18Relation,AO19ProofCplx}.

In this work, we are interested in understanding the relation between
size and degree in Nullstellensatz. In this context
it is relevant to compare and contrast
Nullstellensatz with polynomial calculus as well as with the
well-known \introduceterm{resolution} proof
system~\cite{Blake37Thesis},
which operates directly on the clauses of a CNF formula and repeatedly
derives resolvent clauses
$\clc \lor \cld$
from clauses of the form
$\clc \lor \varx$
and
$\cld \lor \olnot{\varx}$
until contradiction, in the form of the empty clause without any
literals, is reached. For resolution, size is measured by counting the
number of clauses, and \introduceterm{width},  measured as the number
of literals in a largest clause in a refutation, plays an analogous
role to degree for Nullstellensatz and polynomial calculus.

By way of background,
it is not hard to show that for all three proof systems upper bounds
on degree/width imply upper bounds on size, in the sense that if a CNF
formula over $n$~variables can be refuted in degree/width~$d$, then
such a refutation can be carried out in size~$n^{\bigoh{d}}$.
Furthermore, this upper bound has been proven to be tight up to
constant factors in the exponent for resolution and polynomial
calculus~\cite{ALN16NarrowProofs}, and it follows
from~\cite{LLMO09Expressing} that this also holds for Nullstellensatz.
In the other direction, it has been shown for resolution and polynomial
calculus that strong enough lower bounds on degree/width imply lower
bounds on size~\cite{IPS99LowerBounds,BW01ShortProofs}.
This is known to be false for Nullstellensatz, and the pebbling
formulas discussed in more detail later in this paper provide a
counter-example~\cite{BCIP02Homogenization}.

The size lower bounds in terms of degree/width
in~\cite{IPS99LowerBounds,BW01ShortProofs} can be established by
transforming refutations in small size to refutations in small
degree/width. This procedure blows up the size of the refutations
exponentially, however. It is natural to ask whether such a blow-up is
necessary or whether it is just an artifact of the proof. More generally,
given that a formula has proofs in small size and small degree/width,
it is an interesting question whether both measures can be optimized
simultaneously, or whether there has to be a trade-off between the two.

For resolution this question was finally answered
in~\cite{Thapen16Trade-off}, which established that there are indeed
strong trade-offs between size and width making the size blow-up in
\cite{BW01ShortProofs} unavoidable. For polynomial calculus, the
analogous question remains open.

In this paper, we show that there are strong trade-offs between size
and degree for Nullstellensatz.  We do so by establishing a tight
relation between Nullstellensatz refutations of pebbling formulas and
reversible pebblings of the graphs underlying such formulas.  In order
to discuss this connection in more detail, we first need to describe
what reversible pebblings are. This brings us to our next topic.

\subsection{Pebble Games}

In the \introduceterm{pebble game} first studied by Paterson and
Hewitt~\cite{PH70Comparative},
one places pebbles on the vertices of a directed acyclic graph (DAG)~$G$
according to the following rules:
\begin{itemize}
\item
If all (immediate) predecessors of an empty vertex~$v$ contain
  pebbles, a pebble may be placed on~$v$.

\item

  A pebble may be removed from any vertex at any time.
\end{itemize}
The game starts and ends with the graph being empty, and a pebble
should be placed on the (unique) sink of~$G$ at some point.
The complexity measures to minimize are the
total number of pebbles on~$G$ at any given time (the
\introduceterm{pebbling space}) and the number of moves
(the \introduceterm{pebbling time}).

The pebble game has been used to study
flowcharts and recursive schemata~\cite{PH70Comparative},
register allocation~\cite{Sethi75CompleteRegisterAllocation},
time and space as Turing-machine
resources~\cite{Cook74ObservationTimeStorageTradeOff,HPV77TimeVsSpace},
and algorithmic time-space trade-offs
\cite{Chandra73Efficient,
  SS77SpaceTimeFFT,
  SS79SpaceTimeOblivious,
  SS83SpaceTimeLinear,
  Tompa78TimeSpaceComputing}.
In the last two decades,
pebble games have seen a revival
in the context of proof complexity (see, e.g.,~\cite{Nordstrom13SurveyLMCS}),
and pebbling has also turned out to be useful for
applications in cryptography~\cite{DNW05Pebbling,AS15HighParallel}.
An excellent overview of pebbling up to ca.\ 1980 is given
in~\cite{Pippenger80Pebbling} and some more recent developments are
covered in the upcoming survey~\cite{Nordstrom09PebblingSurveyFTTCS}.

Bennett~\cite{Bennett89TimeSpaceReversible} introduced the
\introduceterm{reversible pebble game} as part of a broader
program~\cite{Bennett73LogicalReversibility}
aimed at eliminating or reducing  energy dissipation during computation.
Reversible pebbling has also been of interest in the context of
quantum computing. For example, it was noted in \cite{Soeken1}
that reversible pebble games can be used to capture the problem of
``uncomputing'' intermediate values in quantum algorithms.

The reversible pebble game adds the requirement that the whole
pebbling performed in reverse order should also be a correct pebbling,
which means that the rules for pebble placement and removal become
symmetric as follows:
\begin{itemize}

\item
  If all predecessors of an
empty vertex~$v$   contain pebbles,
  a pebble may  be placed on~$v$.

\item
  If all predecessors of a pebbled vertex~$v$
contain pebbles,
  the pebble on~$v$ may be removed.

\end{itemize}
Reversible pebblings have been studied in
\cite{LV96Reversibility,Kralovic04TimeSpaceReversible,KSS18PebblingMeetsColoring}
and have been used to prove time-space trade-offs in reversible
simulation of irreversible computation in
\cite{LTV98Reversible,LMT00Reversible,Williams00Reversible,BTV01Reversible}.
In a different context, Potechin~\cite{Potechin10Monotone}
implicitly used reversible pebbling to obtain lower bounds in
monotone space complexity, with the connection made explicit in later works
\cite{CP14Monotone,FPRC13Average}.
The paper~\cite{CLNV15Hardness} (to which this overview is indebted)
studied the relative power of standard and reversible pebblings with
respect to space, and also established \PSPACE-hardness results for
estimating the minimum space required to pebble graphs (reversibly or not).

\subsection{Our Contributions}

In this paper, we obtain an exactly tight correspondence between on
the one hand reversible pebblings of DAGs and on the other hand
Nullstellensatz refutations of pebbling formulas over these DAGs.
We show that for any DAG~$G$ it holds that $G$ can be reversibly pebbled
in time~$t$ and space~$s$ if and only if there is a Nullstellensatz
refutation of the pebbling formula over~$G$ in size $t+1$ and
degree~$s$. This correspondence holds regardless of the field in which
the Nullstellensatz refutation is operating, and so, in particular, it
follows that pebbling formulas have exactly the same complexity for
Nullstellensatz regardless of the ambient field.

We then revisit the time-space trade-off literature for the standard
pebble game, focusing on the papers
\cite{CS80GraphPebbling,CS82ExtremeTimeSpaceTradeoffs,LT82AsymptoticallyTightBounds}.
The results in these papers do not immediately transfer to the
reversible pebble game, and we are not fully able to match the
tightness of the results for standard pebbling, but we nevertheless
obtain strong time-space trade-off results for the reversible pebble game.

This allows us to derive Nullstellensatz size-degree trade-offs from
reversible pebbling time-space trade-offs as follows.
Suppose that we have a DAG $G$ such that:
\begin{enumerate}
\item
  $G$ can be reversibly pebbled in time $t_1 \ll t_2$.
\item
  $G$ can be reversibly pebbled in space $s_1 \ll s_2$.
\item
  There is no reversible pebbling of~$G$ that simultaneously achieves
  time $t_1$ and space $s_1$.
\end{enumerate}
Then for Nullstellensatz refutations of the  pebbling
formula
$\pebcontr[G]{}$
over~$G$
(which will be formally defined shortly)
we can deduce that:
\begin{enumerate}
\item
  Nullstellensatz can refute
  $\pebcontr[G]{}$
  in size $t_1 + 1 \ll t_2+1 $.
\item
  Nullstellensatz can also refute
  $\pebcontr[G]{}$
  in degree $s_1 \ll s_2$.
\item
  There is no Nullstellensatz refutation of
  $\pebcontr[G]{}$
  that simultaneously achieves
  size $t_1 + 1$ and degree~$s_1$.
\end{enumerate}

We prove four such trade-off results, which can be found in Section~\ref{sec:revpeb}. The following theorem is one example of such a result (specifically, it is a simplified version of Theorem~\ref{th:NS-exp-trade-off-CS}).

\begin{theorem}
  \label{th:NS-exp-trade-off-example}
  There is a family of $3$-CNF formulas $\set{F_n}_{n=1}^{\infty}$
  of size $\bigtheta{n}$
  \suchthat:
  \begin{enumerate}
  \item There is a Nullstellensatz refutation of $F_n$ in degree
    $s_1 = \Bigoh{\sqrt[6]{n}\log n}$.
  \item
There is a Nullstellensatz refutation of $F_n$ of
    near-linear size
and degree
    $s_2 = \Bigoh{\sqrt[3]{n}\log n}.$
  \item
Any Nullstellensatz refutation of $F_n$ in
    degree at most $\sqrt[3]{n}$
must have exponential size.
\end{enumerate}
\end{theorem}

It should be noted that
this is not the first time proof complexity trade-off results have been
obtained from pebble games. Pebbling formulas were used in
\cite{Ben-Sasson09SizeSpaceTradeoffs,BN11UnderstandingSpace}
to obtain size-space trade-offs for resolution, and later in
\cite{BNT13SomeTradeoffs}
also for polynomial calculus. However, the current reductions between
pebbling and Nullstellensatz are much stronger in that they go in both
directions and are exact even up to additive constants.

With regard to Nullstellensatz, it was shown
in~\cite{BCIP02Homogenization} that Nullstellensatz degree is
lower-bounded by standard pebbling price.  This was strengthened in
\cite{dRMNPRV19Lifting}, which used the connection between designs and
Nullstellensatz degree discussed above to establish that the degree
needed to refute a pebbling formula exactly coincides with the
reversible pebbling price of the corresponding DAG (which is always at
least the standard pebbling price, but can be much larger).
Our reduction significantly improves on~\cite{dRMNPRV19Lifting} by
constructing a more direct reduction, inspired
by~\cite{GKRS18Adventures}, that can simultaneously capture both time
and space.

\subsection{Outline of This Paper}

After having discussed the necessary preliminaries in
\refsec{sec:prelims},
we \ifthenelse{\boolean{conferenceversion}}{prove}{present} the reductions between Nullstellensatz and reversible
pebblings in
\refsec{sec:pebbling-nullstellensatz-correspondence}.
\ifthenelse{\boolean{conferenceversion}}{In
\refsec{sec:revpeb},
we present the size-degree trade-offs for Nullstellensatz we obtain
for different degree regimes.}{In
\refsec{sec:revpeb},
we prove time-space trade-offs for reversible pebblings in order to
obtain size-degree trade-offs for Nullstellensatz.}
\Refsec{sec:conclusion} contains some concluding remarks with
suggestions for future directions of research.

\section{Preliminaries}
\label{sec:prelims}

All logarithms in this paper are base $2$ unless otherwise specified.
For a positive integer~$n$ we write $[n]$ to  denote the set
of integers $\set{1, 2, \ldots, n}$.

A \introduceterm{literal}
$\lita$
over a Boolean variable $\varx$ is either
the variable $\varx$ itself or 
its negation~$\olnot{\varx}$ 
(a \introduceterm{positive} or \introduceterm{negative} literal,
respectively). A \introduceterm{clause}   
$\clc = \lita_1 \lor \formuladots \lor \lita_{\clwidth}$ 
is a disjunction of literals. 
A \introduceterm{$\clwidth$\nobreakdash-clause} is a clause that
contains at most $\clwidth$~literals. A  formula~$\formf$ in
\introduceterm{conjunctive normal form (CNF)}
is a conjunction of
clauses
$\formf = \clc_1 \land \formuladots \land \clc_m$.  
\mbox{A \introduceterm{\kcnfform{}}} is a CNF formula consisting of
\xclause{\clwidth}{}s. We think of clauses and CNF formulas as sets, so that
the order of elements is irrelevant and there are no repetitions.
A truth value assignment~$\partassign$ to the variables of a CNF
formula~$\formf$ is satisfying if every clause in~$\formf$ contains a
literal that is true under~$\partassign$.
    
\subsection{Nullstellensatz}

Let $\F$ be any field and let
$\varvecx = \set{\varx_1, \ldots, \varx_n}$
be a set of variables. We identify a set of polynomials
$\polyset = \setdescr{\polyp_i(\varvecx)}{i \in [m]}$ in the ring~$\F[\varvecx]$ 
with the statement that 
all
$\polyp_i(\varvecx)$
have a common $\set{0,1}$-valued root.
A 
\introduceterm{Nullstellensatz refutation} 
of this claim
is a syntactic equality
\begin{equation}
  \label{eq:ns-refutation-prelims}
  \sum_{i=1}^{m} \polyr_i(\varvecx) \cdot \polyp_i(\varvecx)  
  + 
  \sum_{j=1}^{n} \polys_j(\varvecx)  \cdot (\varx_j^2 - \varx_j) 
  = 1  
\eqcomma
\end{equation}
where 
$\polyr_i,\polys_j$
are also polynomials in~$\F[\varvecx]$. We sometimes refer to the polynomials $\polyp_i(\varvecx)$ as axioms
and $(\varx_j^2 - \varx_j) $ as Boolean axioms.

As discussed in the introduction, Nullstellensatz can be used as a
proof system for CNF formulas by translating a clause
$
\clc  = 
\Lor_{\varx \in P} \varx 
\lor 
\Lor_{\vary \in N} \olnot{\vary} 
$
to the polynomial
$
\polyp(\clc)  = 
\prod_{\varx \in P} (1- \varx)
\cdot 
\prod_{\vary \in N} \vary
$
and viewing Nullstellensatz refutations of
$\setdescr{\polyp({\clc_i})}{i \in [m]}$
as refutations of the CNF formula
$\formf = \Land_{i=1}^{m} \clc_i$.

The 
\introduceterm{degree}
of a Nullstellensatz refutation~\refeq{eq:ns-refutation} 
is
$
\maxofexpr{
  \pdegree{\polyr_i(\varvecx) \cdot \polyp_i(\varvecx)},
  \pdegree{\polys_j(\varvecx) \cdot (\varx_j^2 - \varx_j)}
}
$.
We define the \introduceterm{size} of a refutation
\refeq{eq:ns-refutation-prelims}
to be the total number of monomials encountered when all products of
polynomials are expanded out as linear combinations of monomials.
To be more precise, let
$\monsize{\polyp}$
denote the number of monomials in a polynomial~$\polyp$ written as a linear
combination of monomials.
Then the size of a Nullstellensatz refutation on the
form~\refeq{eq:ns-refutation} is 
\begin{equation}
  \label{eq:ns-refutation-size}
  \sum_{i=1}^{m}
  \Monsize{\polyr_i(\varvecx)} \cdot \Monsize{\polyp_i(\varvecx)}  
  + 
  \sum_{j=1}^{n} 
  2 \cdot
  \Monsize{\polys_j(\varvecx)}  
\eqperiod
\end{equation}
This is consistent with how size is defined for the
``dynamic version'' of Nullstellensatz known as polynomial
calculus~\cite{CEI96Groebner,ABRW02SpaceComplexity}, and also with the general
size definitions for so-called algebraic and semialgebraic proof
systems in
\cite{ALN16NarrowProofs,Berkholz18Relation,AO19ProofCplx}.

We remark that this is not the only possible way of measuring
size, however. It can be noted that the
definition~\refeq{eq:ns-refutation-size} is quite wasteful in that it
forces us to represent the proof in a very inefficient way.
Other papers in the semialgebraic proof complexity literature, such as
\cite{GHP02ExponentialLowerBound,KI06LowerBounds,DMR09TightRankLowerBounds},
instead define size in terms of the polynomials in isolation,
more along the lines of
\begin{equation}
  \label{eq:ns-refutation-size-alt}
  \sum_{i=1}^{m}
  \bigl(
  \Monsize{\polyr_i(\varvecx)} + \Monsize{\polyp_i(\varvecx)}  
  \bigr)
  + 
  \sum_{j=1}^{n} 
  \bigl(
  \Monsize{\polys_j(\varvecx)}  
  + 2
  \bigr) \eqcomma
\end{equation}
or as the bit size or ``any reasonable size'' of the representation of
all polynomials
$
\polyr_i(\varvecx),
\polyp_i(\varvecx)$, and~$\polys_j(\varvecx)
$.

In the end, the difference is not too important since the two measures
\refeq{eq:ns-refutation-size}
and
\refeq{eq:ns-refutation-size-alt}
are at most a square apart, and for size we typically want to
distinguish between polynomial and superpolynomial.
In addition, and more importantly, in this paper we will only deal with
$k$-CNF formulas with $k = \bigoh{1}$, and in this setting the two
definitions are the same up to a constant factor $2^k$.
Therefore, we will stick with~\refeq{eq:ns-refutation-size},
which matches best how size is measured in the closely related
proof systems resolution and polynomial calculus, and which gives the
cleanest statements of our results.\footnote{We refer the reader to
Section 2.4 in~\cite{AH18SizeDegree} for a more detailed discussion of the definition of proof size in algebraic and semialgebraic proof systems.}

When proving lower bounds for algebraic proof systems it is often
convenient to consider a \emph{multilinear} setting where refutations are
presented in the ring
$
\F[\varvecx]
/
\setdescr{\varx_j^2 - \varx_j}{j \in [n]}
$. 
Since the Boolean axioms
$\varx_j^2 - \varx_j$
are no longer needed, the refutation~\refeq{eq:ns-refutation-prelims}
can be written simply as
\begin{equation}
  \label{eq:ns-refutation-ml}
  \sum_{i=1}^{m} \polyr_i(\varvecx) \cdot \polyp_i(\varvecx)  
  = 1  
\eqcomma
\end{equation}
where we assume that all results of multiplications are implicitly
multilinearized. 
It is clear that any refutation on the
form~\refeq{eq:ns-refutation-prelims} remains valid after
multilinearization, and so the size and degree measures can only
decrease in a multilinear setting.
In this paper, we prove our lower bound in our reduction in the
multilinear setting and the upper bound in the non-multilinear
setting, making the tightly matching results even stronger.

\subsection{Reversible Pebbling and Pebbling Formulas}

Throughout this paper
$G = (V,E)$
denotes a directed acyclic graph (DAG) of constant fan-in
with vertices~$\vertices{G}=V$ and edges
$\edges{G}=E$.
For an edge
$(u,v) \in E$
we say that $u$ is a
\introduceterm{predecessor} of $v$ and $v$~a
\introduceterm{successor} of~$u$.
We write $\prednode[G]{v}$ 
to denote the sets of all
predecessors of~$v$, and drop the
subscript when the DAG is clear from context. 
Vertices with no predecessors/successors are called
\introduceterm{sources}/\introduceterm{sinks}.
Unless stated otherwise we will assume that all DAGs under
consideration have a unique sink~$z$.

A
\introduceterm{pebble configuration}
on a DAG
$G = (V,E)$
is a subset of vertices
$\pconf \subseteq V$.
A 
\introduceterm{reversible pebbling strategy} 
for a DAG~$G$ with sink~$z$, or a 
\introduceterm{reversible pebbling} of $G$ for short, 
is a sequence of pebble configurations
$\pebbling = (\pconf_0, \pconf_1,\ldots, \pconf_{\stoptime})$
such that
$\pconf_0 = \pconf_{\stoptime} = \emptyset$,
$z \in \Union_{0\leq t \leq \stoptime} \pconf_t$,
and such that each configuration can be obtained from the previous one
by one of the following rules:
\begin{enumerate}
\item 
  \label{item:placement}
  $\pconfafter=\pconfbefore\unionSP \set{v}$ 
  for
  $v \notin \pconfbefore$
  such that
  $\prednode[G]{v} \subseteq \pconfbefore$
  (a \introduceterm{pebble placement} on~$v$).
\item
  \label{item:reversible-removal}
  $\pconfafter = \pconfbefore\setminus \set{v} $ for
  $ v\in \pconfbefore$
  such that
  $\prednode[G]{v} \subseteq \pconfbefore$
  (a 
  \introduceterm{pebble removal} 
from~$v$).
\end{enumerate}
The \introduceterm{time} of a pebbling
$\pebbling = (\pconf_0, \ldots, \pconf_{\stoptime})$
is $\pebtime{\pebbling} = \stoptime$
and the \introduceterm{space} is
$\pebspace{\pebbling} =
\maxofexpr[0 \leq t \leq {\stoptime}]{\setsize{\pconf_t}}$.

We could also say that a reversible pebbling 
$\pebbling = (\pconf_0, \ldots, \pconf_{\stoptime})$
should be such that 
$\pconf_0 = \emptyset$
and
$z \in \pconf_{\stoptime}$,
and define the time of such a pebbling to be~$2\stoptime$. This is so
since once we have reached a configuration containing~$z$ we can
simply run the pebbling backwards (because of reversibility) until we
reach the empty configuration again,
and without loss of generality all time- and space-optimal reversible
pebblings can be turned into such pebblings.
For simplicity, we will often take this viewpoint in what follows.

\ifthenelse{\boolean{conferenceversion}}{}{
For technical reasons it is sometimes important to
distinguish between  \introduceterm{visiting pebblings}, for which
$z \in \pconf_{\stoptime}$, and
\introduceterm{persistent pebblings}, which meet the more stringent
requirement that
$z \in \pconf_{\stoptime} = \set{z}$.
(It can be noted that for the
more relaxed standard pebble game discussed in the introductory
section any pebbling can be assumed to be persistent without loss of
generality.)}

\providecommand{\pospebvar}[1]{\varx_{#1}}
\providecommand{\negpebvar}[1]{\olnot{\varx}_{#1}}

\providecommand{\pospebvar}[1]{\varx_{#1}}
\providecommand{\negpebvar}[1]{\olnot{\varx}_{#1}}

\begin{figure}[tp]
  \vspace{-7mm}
  \begin{subfigure}[b]{0.28\textwidth}
    \centering
    \includegraphics{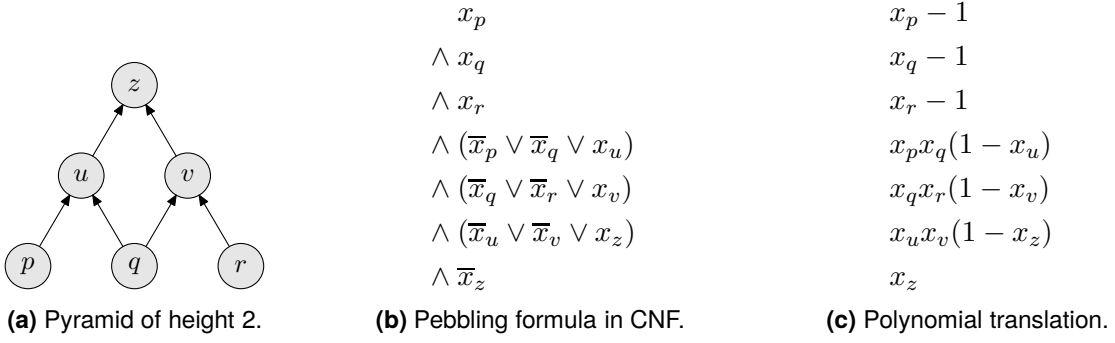}\caption{Pyramid 
of height 2.}
    \label{fig:pebbling-contradiction-for-Pi-2-graph}
  \end{subfigure}
  \hfill
  \begin{subfigure}[b]{0.35\textwidth}
    \centering
    \begin{gather*}
      \begin{aligned}
        & 
        \pospebvar{p}
        \\
        \land \ 
        &
        \pospebvar{q}
        \\
        \land \ 
        &
        \pospebvar{r}
        \\
        \land \ 
        &(\negpebvar{p} \lor \negpebvar{q} \lor \pospebvar{u})
        \\
        \land \ 
        &(\negpebvar{q} \lor \negpebvar{r} \lor \pospebvar{v})
        \\
        \land \ 
        &(\negpebvar{u} \lor \negpebvar{v} \lor \pospebvar{z})
        \\
        \land \ 
        &\negpebvar{z}
      \end{aligned}
    \end{gather*}
    \vspace{-5mm}
    \caption{Pebbling 
formula
in CNF.}	    
    \label{fig:pebbling-contradiction-for-Pi-2-Peb1}
  \end{subfigure}
  \hfill
  \begin{subfigure}[b]{0.35\textwidth}
    \centering
    \begin{gather*}
      \begin{aligned}
        & 
        \pospebvar{p} - 1
        \\
        &
        \pospebvar{q} - 1
        \\
        &
        \pospebvar{r} - 1
        \\
        &\pospebvar{p}  \pospebvar{q} (1 - \pospebvar{u})
        \\
        &\pospebvar{q}  \pospebvar{r} (1 - \pospebvar{v})
        \\
        &\pospebvar{u}  \pospebvar{v} (1 - \pospebvar{z})
        \\
        &\pospebvar{z}
      \end{aligned}
    \end{gather*}
    \vspace{-5mm}
    \caption{Polynomial translation.}	    
    \label{fig:pebbling-contradiction-for-Pi-2-Peb1-poly}
  \end{subfigure}
  \caption{Example pebbling contradiction for the 
    pyramid graph of height~$2$.}
\label{fig:pebbling-contradiction-for-Pi-2}
\end{figure}

Pebble games can be encoded in CNF by so-called
\introduceterm{pebbling formulas}~\cite{BW01ShortProofs},
or
\introduceterm{pebbling contradictions}.
Given a DAG~$G=(V,E)$ with a single sink~$z$,
we associate a variable~$\pospebvar{v}$ with every vertex~$v$ and
add clauses encoding that
\begin{itemize}
\item 
  the source vertices are all true;
\item 
  if all immediate predecessors are true, then truth propagates to the
  successor;
\item 
  but the sink  is false.
\end{itemize}
In short, the pebbling formula over $G$ consists of the clauses $\pospebvar{v} \lor \bigvee_{u \in \prednode[]{v}} \lnot \pospebvar{u}$ for all $v\in V$ (note that if $v$ is a source $\prednode[]{v}=\emptyset$), and the clause $\lnot \pospebvar{z}$.

We encode this formula by
a set of polynomials in the standard way. Given a set $U\subseteq V$, we denote by~$x_{U}$
the monomial $\prod_{u\in U}x_{u}$ (in particular, $x_{\emptyset}=1$).
For every vertex $v\in V$, we have the polynomial
\begin{equation}
A_{v}:=(1-x_{v})\cdot x_{\pred(v)} \eqcomma
\end{equation}
and for the sink~$z$ we also have the polynomial
\begin{equation}
A_{\sink}:=x_{z} \eqperiod
\end{equation}
See
\reffig{fig:pebbling-contradiction-for-Pi-2}
for an illustration, including how the CNF formula is translated to a
set of polynomials.

\section{Reversible Pebblings and Nullstellensatz Refutations}
\label{sec:pebbling-nullstellensatz-correspondence}

In this section, we prove the correspondence between the reversible
pebbling game on a graph~$G$ and Nullstellensatz refutation of the
pebbling contradiction of~$G$. Specifically, we prove the following
result.

\begin{theorem}
\label{thm:pebbling-nullstellensatz-correspondence}Let $G$ be a
directed acyclic graph with a single sink, let $\phi$~be the corresponding
pebbling contradiction, and let $\bF$ be a field. Then, there is
a reversible pebbling strategy for~$G$ with time at most~$t$ and
space at most~$s$ if and only if there is a Nullstellensatz refutation
for~$\phi$ over~$\bF$ of size at most~$t+1$ and degree at most~$s$.
Moreover, the same holds for multilinear Nullstellensatz refutations.
\end{theorem}
We prove each of the directions of Theorem~\ref{thm:pebbling-nullstellensatz-correspondence}
separately in Sections \ref{sub:pebbling-to-refutation} and~\ref{sub:refutation-to-pebbling}
below.

\subsection{\label{sub:pebbling-to-refutation}From Pebbling to Refutation}

We start by proving the ``only if'' direction of Theorem~\ref{thm:pebbling-nullstellensatz-correspondence}.
Let
\begin{equation}
\bP=\left(\bP_{0},\ldots,\bP_{t}\right)
\end{equation}
be an optimal reversible pebbling strategy for~$G$. Let $\bP_{t'}$
be the first configuration in which there is a pebble on the sink~$z$.
Without loss of generality, we may assume that $t=2\cdot t'$: if
the last $t-t'$~steps were more efficient than the first $t'$~steps,
we could have obtained a more efficient strategy by replacing the
first $t'$~steps with the (reverse of) the last $t-t'$~steps,
and vice versa.

We use $\bP$ to construct a Nullstellensatz refutation over~$\bF$
for the pebbling contradiction~$\phi$. To this end, we will first
construct for each step~$i\in\left[t'\right]$ of~$\bP$ a Nullstellensatz
derivation of the polynomial~$x_{\bP_{i-1}}-x_{\bP_{i}}$. The sum
of all these polynomials is a telescoping sum, and is therefore equal
to
\begin{equation}
x_{\bP_{0}}-x_{\bP_{t'}}=1-x_{\bP_{t'}} \eqperiod
\end{equation}
We will then transform this sum into a Nullstellensatz refutation
by adding the polynomial 
\begin{equation}
x_{\bP_{t'}}=A_{\sink}\cdot x_{\bP_{t'}-\left\{ z\right\} } \eqperiod
\end{equation}

We turn to constructing the aforementioned derivations. To this end,
for every $i\in\left[t'\right]$, let $v_{i}\in V$ denote the vertex
which was pebbled or unpebbled during the $i$-th step, i.e., during
the transition from~$\bP_{i-1}$ to~$\bP_{i}$. Then, we know that
in both configurations $\bP_{i-1}$ and $\bP_{i}$ the predecessors
of~$v_{i}$ have pebbles on them, i.e., $\pred(v)\subseteq\bP_{i-1},\bP_{i}$.
Let us denote by $R_{i}=\bP_{i}-\left\{ v_{i}\right\} -\pred(v_{i})$
the set of other vertices that have pebbles during the $i$-th step.
Finally, let $p_{i}$~be a number that equals to~$1$ if $v_{i}$
was pebbled during the $i$-th step, and equals to~$-1$ if $v_{i}$
was unpebbled. Now, observe that
\begin{equation}
x_{\bP_{i-1}}-x_{\bP_{i}}= p_i \cdot x_{\bP_{i-1}}(1 - x_{v_i}) = p_i \cdot x_{R_i}A_{v_{i}} \eqcomma
\label{eq:single-step-derivation}
\end{equation}
where the last step follows since the predecessors of $v_i$ are necessarily in $\bP_{i-1}$.
Therefore, our final refutation for~$\phi$ is
\begin{align}
\nonumber
  \sum_{i=1}^{t'}A_{v_{i}}\cdot p_i \cdot x_{R_{i}}+A_{\sink}\cdot x_{\bP_{t'}-\left\{ z\right\} } & = x_{\bP_{t'}} + \sum_{i=1}^{t'} x_{\bP_{i-1}} - x_{\bP_i}\\
                                                                                                   & = x_{\bP_{t'}} + (x_{\bP_0} - x_{\bP_{t'}}) \\
                                                                                                   & = x_{\bP_{t'}} + (1 - x_{\bP_{t'}}) = 1 \eqperiod
\nonumber
\end{align}
Note, in fact, it is a multilinear Nullstellensatz
refutation, since it contains only multilinear monomials and does
not use the Boolean axioms. It remains to analyze its degree and size.

For the degree, observe that every monomial in the proof is of the form $x_{\bP_i}$, and the degree of each such monomial is exactly the number of pebbles used in the corresponding configuration.
 It follows that the maximal degree is exactly the space of the pebbling strategy~$\bP$.

Let us turn to considering the size. Observe
that for each of the configurations $\bP_{1},\ldots,\bP_{t'}$, the
refutation contains exactly two monomials: for all $i\in [t'-1]$, one monomial for~$\bP_{i}$
is generated in the $i$-th step, and another in the $\left(i+1\right)$-th
step, and for the configuration $\bP_{t'}$
the second
monomial is generated when we add $A_{\sink}\cdot x_{\bP_{t'}-\left\{ z\right\} }$. In addition, the refutation contains exactly one monomial for the
configuration~$\bP_{0}$, which is generated in the first step. Hence,
the total number of monomials generated in the refutation is exactly
$2\cdot t'+1=t+1$, as required.

\subsection{\label{sub:refutation-to-pebbling}From Refutation to Pebbling}

We turn to prove the ``if'' direction of Theorem~\ref{thm:pebbling-nullstellensatz-correspondence}.
We note that it suffices to prove it for multilinear Nullstellensatz
refutations, since every standard Nullstellensatz refutation implies
the existence of a multilinear one with at most the same size and degree. 
Let
\begin{equation}
\sum_{v\in V}A_{v}\cdot Q_{v}+A_{\sink}\cdot Q_{\sink}=1\label{eq:refutation}
\end{equation}
be a multilinear Nullstellensatz refutation of~$\phi$ over~$\bF$
of degree~$s$. We use this refutation to construct a reversible
pebbling strategy~$\bP$ for~$G$.

To this end, we construct a ``configuration graph''~$\bC$, whose
vertices consist of all possible configurations of at most~$s$ pebbles
on~$G$ (i.e., the vertices will be all subsets of~$V$ of size
at most~$s$). The edges of~$\bC$ will be determined by the polynomials~$Q_{v}$
of the refutation, and every edge $\left\{ U_{1},U_{2}\right\} $
in~$\bC$ will constitute a legal move in the reversible pebbling
game (i.e., it will be legal to move from~$U_{1}$ to~$U_{2}$ and
vice versa). We will show that $\bC$~contains a path from the empty
configuration~$\emptyset$ to a configuration~$U_{z}$ that contains
the sink~$z$, and our pebbling strategy will be generated by walking
on this path from $\emptyset$ to~$U_{z}$ and back.

The edges of the configuration graph~$\bC$ are defined as follows:
Let $v\in V$ be a vertex of~$G$, and let $q$~be a monomial of
$Q_{v}$ that \emph{does not contain~$x_{v}$}. Let $W\subseteq V$
be the set of vertices such that $q=x_{W}$ (such a set~$W$ exists
since the refutation is multilinear). Then, we put an edge $e_{q}$
in~$\bC$ that connects $W\cup\pred(v)$ and~$W\cup\pred(v)\cup\left\{ v\right\} $
(we allow parallel edges). It is easy to see that the edge~$e_{q}$
connects configurations of size at most~$s$, and that it indeed
constitutes a legal move in the reversible pebbling game. We note
that $\bC$ is a bipartite graph: to see it, note that every edge~$e_{q}$
connects a configuration of an odd size to a configuration of an even
size.

For the sake of the analysis, we assign the edge~$e_{q}$ a weight
in~$\bF$ that is equal to coefficient of~$q$ in~$Q_{v}$. We
define \emph{the weight of a configuration}~$U$ to be the sum of
the weights of all the edges that touch~$U$ (where the addition
is done in the field~$\bF$). We use the following technical claim,
which we prove at the end of this section.
\begin{claim}
\label{claim:weights-of-configurations}Let $U\subseteq V$ be a configuration
in~$\bC$ that does not contain the sink~$z$. If $U$ is empty,
then its weight is~$1$. Otherwise, its weight is~$0$.
\end{claim}
We now show how to construct the required pebbling strategy~$\bP$
for~$G$. To this end, we first prove that there is a path in~$\bC$
from the empty configuration to a configuration that contains the
sink~$z$. Suppose for the sake of contradiction that this is not
the case, and let $\bH$ be the connected component of~$\bC$ that
contains the empty configuration. Our assumption says that none of
the configurations in~$\bH$ contains~$z$.

The connected component~$\bH$ is bipartite since $\bC$~is bipartite.
Without loss of generality, assume that the empty configuration is
in the left-hand side of~$\bH$. Clearly, the sum of the weights
of the configurations on the left-hand side should be equal to the
corresponding sum on the right-hand side, since they are both equal
to the sum of the weights of the edges in~$\bH$. However, the sum
of the weights of the configurations on the right-hand side is~$0$
(since all these weights are~$0$ by Claim~\ref{claim:weights-of-configurations}),
while the sum of the weights of the left-hand side is~$1$ (again,
by Claim~\ref{claim:weights-of-configurations}). We reached a contradiction,
and therefore $\bH$ must contain some configuration~$U_{z}$ that
contains the sink~$z$.

Next, let $\emptyset=\bP_{0},\bP_{1},\ldots,\bP_{t'}=U_{z}$ be a
path from the empty configuration to~$U_{z}$. Our reversible pebbling
strategy for~$G$ is
\begin{equation}
\bP=\left(\bP_{0},\ldots,\bP_{t'-1},\bP_{t'},\bP_{t'-1},\ldots,\bP_{0}\right) \eqperiod
\end{equation}
This is a legal pebbling strategy since, as noted above, every edge
of~$\bC$ constitutes a legal move of the reversible pebbling game.
The strategy~$\bP$ uses space~$s$, since all the configurations
in~$\bC$ contain at most~$s$ pebbles by definition. The time of~$\bP$
is $t=2\cdot t'$. It therefore remains to show that the size of the
Nullstellensatz refutation from Equation~\ref{eq:refutation} is
at least~$t+1$.

To this end, note that every edge~$e_{q}$ in the path corresponds
to some monomial~$q$ in some polynomial~$Q_{v}$. When the monomial~$q$
is multiplied by the axiom $A_{v}$, it generates two monomials in
the proof: the monomial $q\cdot x_{\pred(v)}$ and the monomial $q\cdot x_{\pred(v)}\cdot x_{v}$.
Hence, the Nullstellensatz refutation contains at least~$2\cdot t'$
monomials that correspond to edges from the path. In addition, the
product~$A_{\sink}\cdot Q_{\sink}$ must contains at least one monomial,
since the refutation must use the sink axiom~$A_{\sink}$ (because
$\phi$ without this axiom is not a contradiction). It follows that
the refutation contains at least $2\cdot t'+1=t+1$ monomials, as
required. We conclude this section by proving Claim~\ref{claim:weights-of-configurations}.
\begin{proof}
[Proof of Claim~\ref{claim:weights-of-configurations}.]We start
by introducing some terminology. First, observe that a mo\-no\-mial~$m$
may be generated multiple times in the refutation of Equation~\ref{eq:refutation},
and we refer to each time it is generated as an \emph{occurrence}
of $m$. We say that an occurrence of~$m$ is \emph{generated by
a monomial~$q_{v}$ of~$Q_{v}$} if it is generated by the product~$A_{v}\cdot q_{v}$.
Throughout the proof, we identify a configuration~$U$ with the monomial~$x_{U}$.

We first prove the claim for the non-empty case. Let $U\subseteq V$
be a non-empty configuration. We would like to prove the weight of~$U$
is~$0$. Recall that the weight of~$U$ is the sum of the coefficients
of the occurrences of~$U$ that are generated by monomials $q_{v}$
\emph{that do not contain the corresponding vertex~$v$}. Observe
that Equation~\ref{eq:refutation} implies that the sum of the coefficients
of \emph{all }the occurrences of~$U$ is~$0$: the coefficient of~$U$
on the right-hand side is~$0$, and it must be equal to the coefficient
of~$U$ on the left-hand side, which is the sum of the coefficients
of all the occurrences.

To complete the proof, we argue that every monomial $q_{v}$ that
does contain the vertex~$v$ contributes~$0$ to that sum. Let $q_{v}$
be a monomial of~$Q_{v}$ that contains the vertex~$v$ and generates
an occurrence of~$U$. Let $\alpha$~be the coefficient of~$q$. 
Then, it must hold that
\begin{align}
\nonumber
A_{v}\cdot q_{v} & =x_{\pred(v)}\cdot q_{v}-x_{v}\cdot x_{\pred(v)}\cdot q_{v}\\
 & =x_{\pred(v)}\cdot q_{v}-x_{\pred(v)}\cdot q_{v}\\
 & =\alpha\cdot x_{U}-\alpha\cdot x_{U} \eqcomma
\nonumber
\end{align}
where the second equality holds since we $q_{v}$~contains~$v$
and we are working with a multilinear refutation, and the third equality
holds since we assumed that $q_{v}$ generates an occurrence of~$U$.
It follows that $q_{v}$~generates two occurrences of~$U$, one
with coefficient $\alpha$ and one with coefficient~$-\alpha$, and
therefore it contributes~$0$ to the sum of the coefficients of all\emph{
}the occurrences of~$U$.

We have shown that the sum of the coefficients of all\emph{ }the occurrences
of~$U$ is~$0$, and that the occurrences generated by monomials
$q_{v}$ that contain~$v$ contribute~$0$ to this sum, and therefore
the sum of coefficients of occurrences that are generated by monomials
$q_{v}$ that do not contain\emph{~$v$} must be~$0$, as required.
In the case that $U$ is the empty configuration, the proof is identical,
except that the sum of the coefficients of all occurrences is~$1$,
since the coefficient of~$\emptyset$ is~$1$ on the right hand
side of Equation~\ref{eq:refutation}.
\end{proof}

\section{Nullstellensatz Trade-offs from Reversible Pebbling}
\label{sec:revpeb}

\ifthenelse{\boolean{conferenceversion}}{
In this section we present the Nullstellensatz size-degree trade-offs we obtain for different degree regimes.}{In this section we prove Nullstellensatz refutation size-degree trade-offs for different degree regimes.} Let us first recall what is known with regards to degree and size. In what follows, a Nullstellensatz refutation of a CNF formula $F$ refers to a Nullstellensatz refutation of the translation of $F$ to polynomials.
As mentioned in the introduction, if a CNF formula over $n$ variables can be refuted in degree $d$ then it can be refuted in simultaneous degree $d$ and size $n^{\bigoh{d}}$. However, for Nullstellensatz it is not the case that strong enough degree lower bounds imply size lower bounds. 

A natural question is whether for any given function $d_1(n)$ there is a family of CNF formulas $\set{F_n}_{n=1}^{\infty}$  of size $\bigtheta{n}$ such that
\begin{enumerate}
\item $F_n$ has a Nullstellensatz refutation $d_1(n)$;
\item $F_n$ has a Nullstellensatz refutation of (close to) linear size and degree $d_2(n) \gg d_1(n)$;
\item Any Nullstellensatz refutation of $F_n$ in degree only slightly below $d_2(n)$ must have size nearly~$n^{d_1(n)}$.
\end{enumerate} 

We present explicit constructions of formulas providing such trade-offs in several different parameter regimes. We first show that there are formulas that require exponential size in Nullstellensatz if the degree is bounded by some polynomial function, but if we allow slightly larger degree there is a nearly linear size proof.

\begin{theorem}
  \label{th:NS-exp-trade-off-CS}
  There is a family of explicitly constructible 
  unsatisfiable $3$-CNF formulas $\set{F_n}_{n=1}^{\infty}$ 
  of size $\bigtheta{n}$
  \suchthat:
  \begin{enumerate}
  \item There is a Nullstellensatz refutation of $F_n$ in degree
    $d_1 = \Bigoh{\sqrt[6]{n}\log n}$.    
  \item
    For any constant $\epsilon > 0$, 
    there is a Nullstellensatz refutation of $F_n$ of size
    $\bigoh{n^{1+\epsilon}}$ and degree 
    $d_2 = \Bigoh{d_1 \cdot \sqrt[6]{n}} = \Bigoh{\sqrt[3]{n}\log n}.$
  \item 
    There exists a constant $K >0$ such that 
    any Nullstellensatz refutation of $F_n$ in
    degree at most $d = K d_2/\log n = \Bigoh{\sqrt[3]{n}}$
    must have size
    $\bigl( \sqrt[6]{n} \bigr)! \,$.
  \end{enumerate}
\end{theorem}

We also analyse a family of formulas that can be refuted in close to logarithmic degree and show that even if we allow up to a certain polynomial degree, the Nullstellensatz size required is superpolynomial.

\begin{theorem}\label{th:NS-non-constant-trade-off-CS}
  Let 
  $\delta > 0$ be an arbitrarily small positive constant 
  and let $\csfunc(n)$ be any arbitrarily slowly growing monotone function
  $\littleomega{1} = \csfunc(n) \leq {n}^{1/4}$.
Then there is a family of explicitly constructible 
  unsatisfiable $3$-CNF formulas $\set{F_n}_{n=1}^{\infty}$ 
  of size $\bigtheta{n}$
  \suchthat:
  \begin{enumerate}
  \item 
    There is a Nullstellensatz refutation of $F_n$ in degree
    $d_1 = \csfunc(n)\log(n)$.
\item 
    For any constant $\epsilon > 0$, 
    there is a Nullstellensatz refutation of $F_n$ of size
    $\bigoh{n^{1+\epsilon}}$ and degree 
    \[ d_2 = \Bigoh{d_1\cdot n^{1/2}/g(n)^2} = \Bigoh{n^{1/2}\log n / g(n)}.\]
  \item 
    Any Nullstellensatz refutation of $F_n$ in
    degree at most 
    \[ d = \Bigoh{d_2/n^{\delta}\log n} = \Bigoh{n^{1/2-\delta}/g(n)}\]
    must have size superpolynomial in~$n$.
  \end{enumerate}
\end{theorem}

Still in the small-degree regime, we present a very robust trade-off in the sense that superpolynomial size lower bound holds for degree from $\log^2(n)$ to $n/\log(n)$. 

\begin{theorem}
  \label{th:NS-robust-trade-off-LT}
  There is a family of explicitly constructible 
  unsatisfiable $3$-CNF formulas $\set{F_n}_{n=1}^{\infty}$ 
  of size $\bigtheta{n}$
  \suchthat:
  \begin{enumerate}
  \item 
    There is a Nullstellensatz refutation of $F_n$ in degree
    $d_1 = \bigoh{\log^2 n}$.  
  \item
    For any constant $\delta > 0$, 
    there is a Nullstellensatz refutation of $F_n$ of size
    $\bigoh{n}$ and degree
    \[d_2 = \bigoh{d_1 \cdot n / \log^{3-\delta} n} = \bigoh{n/\log^{1-\delta} n}.\]
  \item 
    There exists a constant $K >0$ such that 
    any Nullstellensatz refutation of $F_n$ in
    degree at most $d = K d_2 /  \log^\delta n = O(n / \log n)$ 
    must have size $n^{\bigomega{\log \log n}}$.
  \end{enumerate}
\end{theorem}

Finally, we study a family of formulas that have Nullstellensatz refutation of quadratic size and that present a smooth size-degree trade-off.

\begin{theorem}\label{th:NS-permutation}
    There is a family of explicitly constructible 
    unsatisfiable $3$-CNF formulas $\set{F_n}_{n=1}^{\infty}$ 
    of size $\bigtheta{n}$  
    such that any
    Nullstellensatz refutation of $F_n$
    that optimizes
    size given degree constraint 
    $d=n^{\bigtheta{1}}$ (and less than $n$) 
    has size
    $
    \Bigtheta{n^2 / d}
    $. 
  \end{theorem}

We prove these results by obtaining the analogous time-space trade-offs for reversible pebbling and then applying the tight correspondence between size and degree in Nullstellensatz and time and space in reversible pebbling.
\ifthenelse{\boolean{conferenceversion}}{We differ the reader to the upcoming full version for the details.}{
\subsection{Reversible Pebbling Time-Space Trade-offs}

Our strategy for proving reversible pebbling trade-offs will be to analyse standard pebbling trade-offs.
Clearly lower bounds from standard pebbling transfer to reversible pebbling; the next theorem shows how, in a limited sense, we can also transfer \emph{upper bounds}.
It is based on a reversible simulation of irreversible computation proposed by~\cite{Bennett89TimeSpaceReversible} and analysed precisely in~\cite{LS90NoteOnBennett}.

\begin{theorem}[\cite{Bennett89TimeSpaceReversible,LS90NoteOnBennett}]\label{th:rev-simulation}
Let $G$ be an arbitrary DAG  and suppose $G$ can be pebbled (in the standard way) 
using $\spacestd$ pebbles in time $t\geq 2\spacestd$. Then for any $\epsilon >0$, $G$ can be reversibly pebbled in time ${t^{1+\epsilon}/\spacestd^{\epsilon}}$ using  $\epsilon (2^{1/\epsilon} -1) \, \spacestd \, {\log (t/\spacestd)}$ pebbles.
\end{theorem}

We also use the following general proposition, which allows upper bounding the reversible pebbling price of a graph by using its depth and maximum in-degree.

\begin{proposition}
  \label{pr:rev-upper-bound-indegree-depth-tight}
  Any DAG with maximum indegree $\indegreedag$ and depth $\depthdag$ has
  a persistent reversible pebbling strategy in space at most $\depthdag \indegreedag + 1$.
\end{proposition}

\begin{proof}
We will use the fact that if a graph has a persistent reversible strategy in space $s$ then it has a visiting reversible strategy in space $s$.

The proof is by induction on the depth. 
For $\depthdag=0$ we can clearly persistently reversibly pebble the graph with $1$ pebble.
For $\depthdag \geq 1$, we persistently reversibly pebble all but one of the (that is, at most $\indegreedag-1$) immediate predecessors of the sink one at a time. By the induction hypothesis, this can be done with at most $\indegreedag - 2 + (\depthdag-1)\indegreedag + 1 = \depthdag \indegreedag - 1$ pebbles. At this point there are at most $\indegreedag-1$ predecessors of the sink which are pebbled and no other pebbles on the graph. Let $v$ be the only non-pebbled predecessor of the sink. We do a visiting reversible pebbling of $v$ until a pebble is placed on $v$. We now pebble the sink and then reverse the visiting pebbling of $v$ until the subtree rooted at $v$ has no pebbles on it. By the induction hypothesis, this can be done with at most $\indegreedag + (\depthdag - 1)\indegreedag + 1 = \depthdag\indegreedag + 1$ pebbles. All that is left to do is to to remove the $\indegreedag - 1$ pebbles which are on predecessors of the sink. Again by the induction hypothesis, this can be done with $\indegreedag + (\depthdag - 1)\indegreedag + 1$ pebbles.
\end{proof}

\subsection{Carlson-Savage Graphs}\label{sec:carlson-savage}

The first family of graphs for which we present reversible pebbling trade-offs consists of the so-called Carlson-Savage graphs, which are illustrated in Figure~\ref{fig:graph-cs-induction} and are defined as follows. 

\begin{definition}[Carlson-Savage graph
  \cite{CS80GraphPebbling, CS82ExtremeTimeSpaceTradeoffs,
    Nordstrom12RelativeStrength}]
  \label{def:CS-dag}
  The two-parameter graph family
  $\csdag{\csnspines}{\csrec}$,
  for
  $\csnspines, \csrec \in \Nplus$,
  is defined
  by induction over $\csrec$.
The base case
  $\csdag{\csnspines}{1}$
  is a DAG consisting of 
  two sources
  $\sourcestd_1,\sourcestd_2$
  and
  $\csnspines$ sinks
  $\cssink[1], \ldots, \cssink[\csnspines]$
  with directed edges
  $(\sourcestd_i, \cssink[j])$,
  for $i=1, 2$ and $j= 1, \ldots, \csnspines$,
  \ie edges from both sources to all sinks.
The graph
  $\csdag{\csnspines}{\csrec + 1}$
  has
  $\csnspines$ sinks and is
  built from the following components:
  \begin{itemize}
  \item
    $\csnspines$ disjoint copies
    $\pyramidgraph[\csrec]^{\graphcopyindex{1}},
    \ldots, 
    \pyramidgraph[\csrec]^{\graphcopyindex{\csnspines}}$
    of a pyramid graph
of height $\csrec$.
    
  \item
    one copy of
    $\csdag{\csnspines}{\csrec}$.
    
  \item
    $\csnspines$ disjoint and identical line graphs called \introduceterm{spines}, where each spine is composed of $ \csrec$ \introduceterm{sections}, and
    every section contains $2 \csnspines$ vertices.
 
  \end{itemize}
  
  The above components are connected as follows: In every section of every spine,
  each of the first $c$ vertices has an incoming edge from the sink of one of the first $c$ pyramids,
  and each of the last $c$ vertices has an incoming edge from one of the sinks of $\csdag{\csnspines}{\csrec}$
 (where different vertices in the same section are connected to different sinks).
   
\end{definition}

\begin{figure}[tp]
  \centering
  \includegraphics{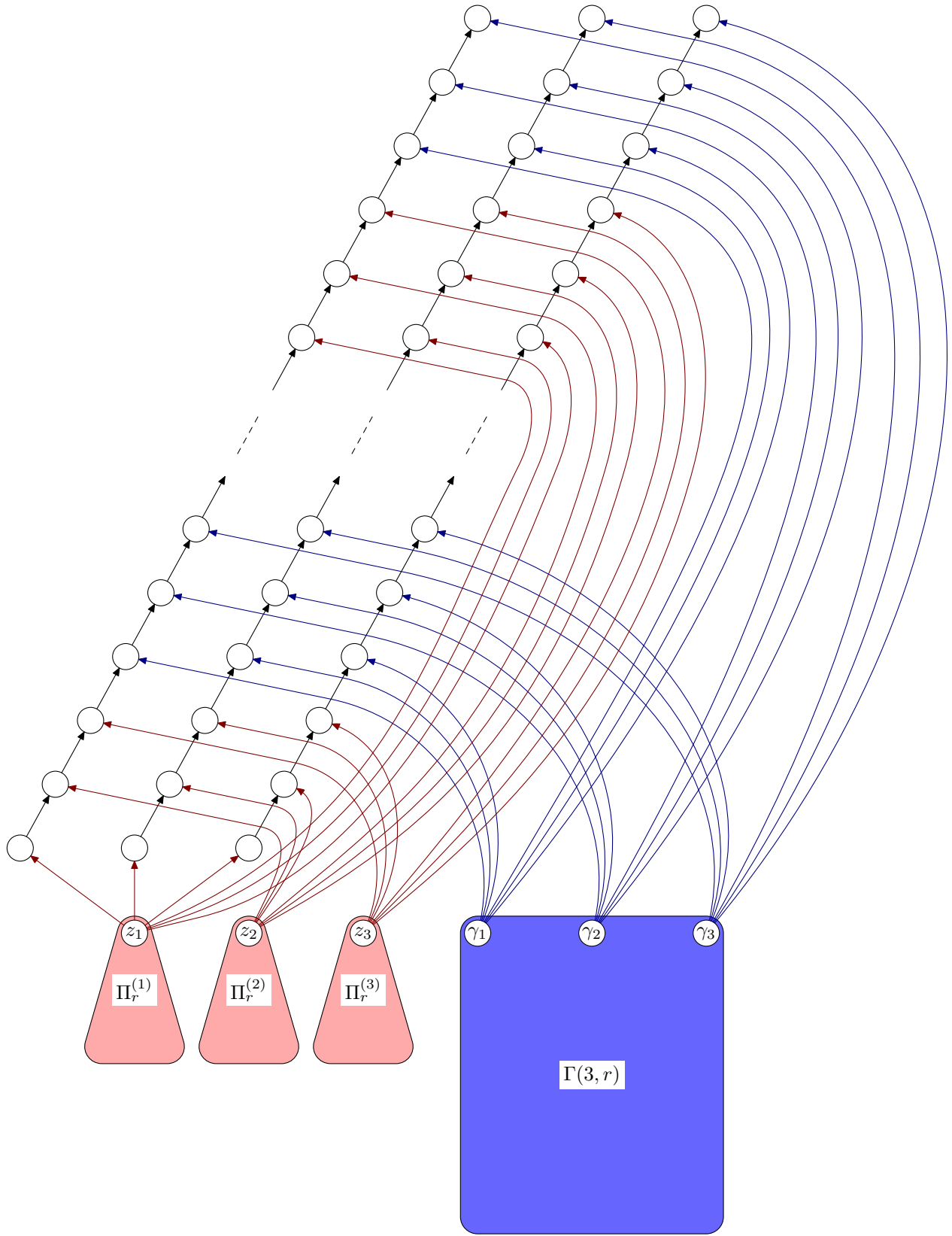}
  \caption{Inductive definition of Carlson-Savage graph $\csdag{3}{\csrec+1}$
    with $3$ spines and sinks.}
  \label{fig:graph-cs-induction}
\end{figure}

Note that $\csdag{\csnspines}{\csrec}$ has multiple sinks. We define a (reversible) pebbling of a multi-sink graph to be a (reversible) pebbling that places pebbles on each sink at some point (the pebbles do not need to be present in the last configuration). 
Let $\Gamma'(\csnspines,\csrec)$ be the single-sink subgraph of 
$\csdag{\csnspines}{\csrec}$ consisting of all vertices 
that reach the first sink of $\csdag{\csnspines}{\csrec}$.
Since all sinks are symmetric, pebbling $\Gamma'(\csnspines,\csrec)$
and pebbling $\csdag{\csnspines}{\csrec}$ are almost equivalent.

\begin{proposition}\label{pr:multi-sink-CS}
Any (reversible) pebbling $\pebbling$ of $\csdag{\csnspines}{\csrec}$
induces a (reversible) pebbling $\pebbling'$ of $\Gamma'(\csnspines,\csrec)$
in at most the same space and the same time.
From any (reversible) pebbling $\pebbling'$ of $\Gamma'(\csnspines,\csrec)$ 
we can obtain (reversible) pebbling $\pebbling$ of $\csdag{\csnspines}{\csrec}$ by (reversibly) pebbling each sink of $\csdag{\csnspines}{\csrec}$ one at a time, that is, simulating $\pebbling'$ $\csnspines$ times, once for each sink. Note that $\pebspace{\pebbling} = \pebspace{\pebbling'}$ and $\pebtime{\pebbling} = \csnspines \cdot \pebtime{\pebbling'}$.
\end{proposition}

Carlson and Savage proved the following properties of this graph.

\begin{lemma}[\cite{CS82ExtremeTimeSpaceTradeoffs}]
  \label{lem:CS-dag-easy-prop}
  The graphs
  $\csdag{\csnspines}{\csrec}$
  are of size
  $
\Bigtheta{\csnspines \csrec^3 + \csnspines^2 \csrec^2}$, have in-degree $2$,
  and have standard pebbling price
$\csrec + 2$.
\end{lemma}

\begin{theorem}[\cite{CS82ExtremeTimeSpaceTradeoffs}]
  \label{th:CS-dag-trade-off-black-only}
  Suppose that $\pebbling$
  is a standard pebbling of
  $\csdag{\csnspines}{\csrec}$ in space less than
$ (\csrec + 2) + \cspebsurplus
  $
  for $0 < \cspebsurplus \leq \csnspines - 3$. 
  Then 
\begin{equation*}
    \pebtime{\pebbling} 
    \geq
    \left(
    \frac{\csnspines -  \cspebsurplus}{\cspebsurplus + 1}
    \right)^{\csrec}
    \cdot
    \csrec!
    \eqperiod
  \end{equation*}
\end{theorem}

This lower bound holds for space up to $\csnspines + \csrec - 1$. By allowing only a constant factor more pebbles it is possible to pebble the graph (in the standard way) in linear time.

\begin{lemma}[\cite{Nordstrom12RelativeStrength}]
  \label{lem:CS-dag-lin-time}
  The graphs
  $\csdag{\csnspines}{\csrec}$
  have standard pebbling strategies 
in simultaneous space 
  $\bigoh{\csnspines + \csrec}$
  and time linear in the size of the graphs.
\end{lemma}

The standard pebbling price upper bound does not carry over
to reversible pebbling because 
the line graph requires more pebbles in reversible pebbling than in standard pebbling.
However, we can adapt the standard pebbling strategy to reversible pebbling using the following fact on the line graph.

\begin{proposition}[\cite{LV96Reversibility}]
  \label{pr:rev-peb-price-path}
   The visiting reversible pebbling price of the line graph on $n$ vertices is
   $\ceiling{\log (n+1) }$ and the persistent reversible pebbling price is
   $\floor{\log (n - 1)} + 2$.
\end{proposition}

Using this result, we get the following upper bound (which is slightly stronger
then what we would get by applying Theorem~\ref{th:rev-simulation}).

\begin{lemma}
  \label{lem:rev-CS-dag-easy-prop}
  The graphs
  $\csdag{\csnspines}{\csrec}$
  have reversible pebbling price
  at most
  $\csrec(\log(\csnspines\csrec) + 3)$.
\end{lemma}

\begin{proof}
The proof is by induction on $\csrec$.
Clearly, $\csdag{\csnspines}{1}$ can be reversibly pebbled
with $3$ pebbles. 
In order to pebble any sink of $\csdag{\csnspines}{\csrec}$,
we can reversibly pebble the corresponding spine with the space-efficient strategy for reversibly pebbling a line graph (as per Proposition~\ref{pr:rev-peb-price-path}) and every time we need to place or remove a pebble on a given vertex of the spine, we reversibly pebble the subgraph that reaches this vertex. 
By Proposition~\ref{pr:rev-upper-bound-indegree-depth-tight}, pyramids of depth $\csrec -1$ can be reversibly pebbled with $2(\csrec -1) + 1$ pebbles.
Therefore, by induction on $\csrec$ we get that the reversible pebbling price of $\csdag{\csnspines}{\csrec}$ is at most 
$\max\set{(\csrec-1)(\log(\csnspines\csrec) + 3), 2(\csrec-1) + 1 } +\log(2\csnspines\csrec) + 2 
\leq (\csrec-1)(\log(\csnspines\csrec) + 3) + \log(\csnspines\csrec) + 3$.
\end{proof}

In order to obtain nearly linear time reversible pebbling, 
we apply Theorem~\ref{th:rev-simulation} to Lemma~\ref{lem:CS-dag-lin-time}.

\begin{lemma}\label{lem:rev-CS-dag-lin-time}
  For any $\epsilon>0$, the graphs
  $\csdag{\csnspines}{\csrec}$
  have reversible pebbling strategies 
in simultaneous space 
  $\bigoh{\epsilon 2^{1/\epsilon}(\csnspines + \csrec)\log(\csnspines \csrec)}$
  and time $\bigoh{n^{1+\epsilon}}$ (where $n$ denotes the number of vertices).
\end{lemma}

We can now choose different values for the parameters $\csnspines$ and $\csrec$ and obtain graphs with trade-offs in different space regimes.
The first family of graphs we consider are those that exhibit exponential time lower-bounds.

\begin{theorem}
  \label{th:exp-trade-off-CS}
  There are  explicitly constructible families of \singlesinkdagtext{}s
  $\set{G_n}_{n=1}^{\infty}$
  of size $\bigtheta{n}$ and maximum in-degree $2$ 
  \suchthat:
  \begin{enumerate}
  \item The graph
    $G_n$ 
    has reversible pebbling price $s_1 = \Bigoh{\sqrt[6]{n}\log n}$.

  \item
    For any constant $\epsilon > 0$, 
    there is a reversible pebbling of $G_n$ in time
    $\bigoh{n^{1+\epsilon}}$ and space \[s_2 = \Bigoh{s_1 \cdot \sqrt[6]{n}} = \Bigoh{\sqrt[3]{n}\log n}\eqperiod\]
\item 
   There is a constant $K > 0$ such that
    any standard pebbling of $G_n$ in
    space at most \[s = \frac{Ks_2}{ \log n} = \Bigoh{\sqrt[3]{n}}\]
    must take time at least
    $\bigl( \sqrt[6]{n} \bigr)! \,$.

  \end{enumerate}
\end{theorem}
\begin{proof}
Let $G_n$ be the single-sink subgraph of 
$ \csdag{\csnspines(n)}{\csrec(n)}$
consisting of all vertices that reach the first sink,
for $\csnspines(n) = \sqrt[3]{n}$
and $\csrec(n) = \sqrt[6]{n}$.
  
By Lemma~\ref{lem:CS-dag-easy-prop}, $G_n$ has $\bigtheta{n}$ vertices and by Proposition~\ref{pr:multi-sink-CS}, items~\ref{item:rev-peb-price}--\ref{item:rev-peb-trade-off} follow from Lemma~\ref{lem:rev-CS-dag-easy-prop}, Lemma~\ref{lem:rev-CS-dag-lin-time} and Theorem~\ref{th:CS-dag-trade-off-black-only}, respectively.
\end{proof}

Given Theorem~\ref{thm:pebbling-nullstellensatz-correspondence} which proves the tight correspondence between reversible pebbling and Nullstellensatz refutations, Theorem~\ref{th:NS-exp-trade-off-CS} follow from Theorem~\ref{th:exp-trade-off-CS}.

It is also interesting to consider families of graphs that
can be reversibly pebbled in very small space, close to the logarithmic lower bound on the number of pebbles required to reversibly pebble a single-sink DAG. 
In this small-space regime, we cannot expect exponential time lower bounds, but we can still obtain superpolynomial ones.

\begin{theorem}\label{th:non-constant-trade-off-CS}
  Let 
  $\delta > 0$ be an arbitrarily small positive constant 
  and let $\csfunc(n)$ be any arbitrarily slowly growing monotone function
  $\littleomega{1} = \csfunc(n) \leq {n}^{1/4}$.
Then there is a family of explicitly constructible \singlesinkdagtext{}s 
  $\set{G_n}_{n=1}^{\infty}$
  of size $\bigtheta{n}$  and maximum in-degree $2$
  \suchthat:
  \begin{enumerate}
  \item \label{item:rev-peb-price} 
    The graph
    $G_n$ 
    has reversible pebbling price 
    $s_1 \leq \csfunc(n)\log(n)$.

\item For any constant $\epsilon>0\,$, 
    there is a reversible pebbling of $G_n$ in time
    $\bigoh{n^{1+\epsilon}}$
    and space \[ s_2 = \Bigoh{s_1\cdot n^{1/2}/g(n)^2} = \Bigoh{n^{1/2}\log n / g(n)}\eqperiod\]

  \item \label{item:rev-peb-trade-off}
    Any standard pebbling of $G_n$
    in space at most \[ s = \Bigoh{s_2/n^{\delta}\log n} = \Bigoh{n^{1/2-\delta}/g(n)}\]
    requires time superpolynomial in~$n$.
    
  \end{enumerate}
\end{theorem}

\begin{proof}
The proof is analogous to that of Theorem~\ref{th:non-constant-trade-off-CS} with parameters $\csrec(n) = \csfunc(n)$  and $
  \csnspines(n) = 
  {{n}^{1/2} / \csfunc(n)}
  $.
\end{proof}

By applying Theorem~\ref{thm:pebbling-nullstellensatz-correspondence} to the above result we obtain Theorem~\ref{th:NS-non-constant-trade-off-CS}.

\begin{remark}

We note that in the second items of both the foregoing theorems, we could have reduced the time of the reversible pebbling to $\bigoh{n^{1 + \littleoh{1}}}$ by applying
Theorem~\ref{th:rev-simulation} with $\epsilon = \BIGOH{\frac{1}{\log \log n}}$. This would have come at a cost of an extra logarithmic factor in the corresponding space bounds.

\end{remark}

\subsection{Stacks of Superconcentrators}

Lengauer and Tarjan~\cite{LT82AsymptoticallyTightBounds} studied robust superpolynomial trade-offs for standard pebbling and showed that there are graphs that have standard pebbling price
$\bigoh{\log^2 n}$, but that any standard pebbling in space up to $K{n/\log n}$, for some constant $K$, requires superpolynomial time. 
For reversible pebbling we get almost the same result for the same family of graphs.

\begin{theorem}
  \label{th:robust-trade-off-LT}
  There are  explicitly constructible families of \singlesinkdagtext{}s
  $\set{G_n}_{n=1}^{\infty}$
  of size $\bigtheta{n}$ and maximum in-degree $2$
  \suchthat:
  \begin{enumerate}
  \item 
    The graph
    $G_n$ 
    has reversible pebbling price
    $s_1 = \bigoh{\log^2 n}$.
    \label{item:1-robust-LT}

  \item For any constant $\delta > 0$,  
    there is a reversible pebbling of $G_n$
    in time $\bigoh{n}$ and space
    \[s_2 = \bigoh{s_1 \cdot n / \log^{3-\delta} n} = \bigoh{n/\log^{1-\delta} n}\eqperiod\]
\label{item:2-robust-LT}
    \vspace{-2em} \item 
    There exists a constant $K>0$ \suchthat any standard pebbling $\pebbling_n$ of $G_n$ using at most pebbles $s = \frac{K s_2} { \log^\delta n} = O(n / \log n)$ requires time $n^{\bigomega{\log \log n}}$.
    \label{item:3-robust-LT}
  \end{enumerate}
\end{theorem}

Note that, together with Theorem~\ref{thm:pebbling-nullstellensatz-correspondence}, this implies Theorem~\ref{th:NS-robust-trade-off-LT}.
Now in order to prove this theorem we must first introduce the family of graphs we will consider.

\begin{definition}[Superconcentrator~\cite{Valiant75NonLinear}]
  \label{def:superconcentrator}
  A directed acyclic graph $G$ is an 
  $\ltsupconcsize$-superconcentrator
  if it has $\ltsupconcsize$ sources 
  $\sourcesetstd = 
  \set{\sourcestd_1, \ldots, \sourcestd_{\ltsupconcsize}}$,
  $\ltsupconcsize$ sinks 
  $\sinksetstd = 
  \set{\sinkstd_1, \ldots, \sinkstd_{\ltsupconcsize}}$,
  and for any subsets 
  $\sourcesetstd'$ 
  and
  $\sinksetstd'$ 
  of sources and sinks of size
  $
  \Setsize{\sourcesetstd'} =
  \Setsize{\sinksetstd'} = \ell$
  it holds that there are $\ell$ vertex-disjoint paths between
  $\sourcesetstd'$ 
  and
  $\sinksetstd'$ 
    in~$G$.
\end{definition}

Pippenger~\cite{Pippenger77Superconcentrators} proved that there are superconcentrators of linear size and logarithmic depth, and Gabber and Galil~\cite{GG81Explicit} gave the first explicit construction.
For concreteness, we will consider the explicit construction by Alon and Capalbo~\cite{AC03SmallerExplicitSuperconcentrators} which has better parameters.

\begin{theorem}[\cite{AC03SmallerExplicitSuperconcentrators}]
For all integers $k\geq 6$, 
there are explicitly constructible 
$\ltsupconcsize$-superconcentrators for 
$\ltsupconcsize=\bigoh{2^k}$
with $\bigoh{\ltsupconcsize}$ edges, depth $\bigoh{\log \ltsupconcsize}$, and maximum indegree $\bigoh{1}$.
\end{theorem}

It is easy to see that we can modify these superconcentrators so that the maximum indegree is $2$ by substituting each vertex with indegree $\delta > 2$ by a binary tree with $\delta$ leafs. Note that this only increase the size and the depth by constant factors.

\begin{corollary}\label{cor:linear-size-sc}
There are $\ltsupconcsize$-superconcentrators with $\bigoh{\ltsupconcsize}$ vertices, maximum indegree $2$ and depth $\bigoh{\log \ltsupconcsize}$.
\end{corollary}

Given an $\ltsupconcsize$-superconcentrator  $G_{\ltsupconcsize}$, we define a stack of $r$ superconcentrators $G_{\ltsupconcsize}$ to be $r$ disjoint copies of 
$G_{\ltsupconcsize}$ where each sink of the $i$th copy is connected to a different source of the $i+1$st copy for $i\in[r-1]$. Since we want single-sink DAGs, we add a binary tree with $\ltsupconcsize$ leafs and depth $\ceiling{\log\ltsupconcsize}$, and connect each sink of the $r$th copy of $G_{\ltsupconcsize}$ to a different leaf of the tree.
Lengauer and Tarjan~\cite{LT82AsymptoticallyTightBounds} proved the following theorem for stacks of superconcentrators.

\begin{theorem}[\cite{LT82AsymptoticallyTightBounds}]
  \label{th:ssc-special-case}
  Let
  $\ltsscdag{\ltsupconcsize}{\ltncopies}$
 denote a stack of $\ltncopies$
  (explicitly constructible)
  linear-size  \nsupconctext{\ltsupconcsize}
  with bounded indegree
  and depth 
  $\log \ltsupconcsize$.
  Then the following holds:
  \begin{enumerate}
  \item
    The standard pebbling price of $\ltsscdag{\ltsupconcsize}{\ltncopies}$ is $\bigoh{\ltncopies \log \ltsupconcsize}$.

  \item \label{item:2-stack-of-sc}
    There is a linear-time standard pebbling strategy $\pebbling$ for
    $\ltsscdag{\ltsupconcsize}{\ltncopies}$
    with
    $\pebspace{\pebbling} = \bigoh{\ltsupconcsize}$.

  \item \label{item:3-stack-of-sc}
    If $\pebbling$ is a standard pebbling strategy for
    $\ltsscdag{\ltsupconcsize}{\ltncopies}$
    in space
    $\spacestd \leq \ltsupconcsize / 20$,
    then
    $
    \pebtime{\pebbling} \geq
    \ltsupconcsize \cdot
    \bigl(
    \frac{\ltncopies \ltsupconcsize}{64 \spacestd}
    \bigr)^{\ltncopies}
    $.
  \end{enumerate}
\end{theorem}

With this result in hand we can now proceed to prove Theorem~\ref{th:robust-trade-off-LT}.

\begin{proof}[Proof of Theorem~\ref{th:robust-trade-off-LT}]
Let $G_n$ be a stack of $\log n$ linear-size $(n/\log n)$-superconcentrators as per Corollary~\ref{cor:linear-size-sc}. Note that $G_n$ has $\bigtheta{n}$ vertices, indegree $2$ and depth $\bigoh{\log^2 n}$. By Proposition~\ref{pr:rev-upper-bound-indegree-depth-tight} we have that $G_n$ can be reversibly pebbled with $\bigoh{\log^2 n}$ pebbles (proving item~\ref{item:1-robust-LT}). 

By item~\ref{item:2-stack-of-sc} in Theorem~\ref{th:ssc-special-case} and by choosing $\epsilon = 1/(\delta \log \log n)$ in Theorem~\ref{th:rev-simulation} we conclude that $G_n$ can be reversibly pebbled in simultaneous time $\Bigoh{n2^{1/\delta}}$ and space $\Bigoh{n/(\delta\log^{1-\delta} n)}$, from which item~\ref{item:2-robust-LT} follows.
Finally, item~\ref{item:3-robust-LT} in the theorem follows from item~\ref{item:3-stack-of-sc} in Theorem~\ref{th:ssc-special-case}.
\end{proof}

\subsection{Permutation Graphs}

Another family of graphs that has been studied in the context of standard pebbling trade-offs is that of permutation graphs. 

\begin{definition}
\label{def:permgraph}
Given a permutation $\sigma\in\mathfrak{S}([n])$, the \introduceterm{permutation graph} $G(\sigma)$ consists of two lines
$(x_1, \ldots, x_n)$ (the \emph{bottom line}) and  $(y_1, \ldots, y_n)$ (\emph{top line}) which are connected as follows:
for every $1 \le i \le n$, there is an edge from $x_i$ to $y_{\sigma(i)}$.
\end{definition}

Lengauer and Tarjan~\cite{LT82AsymptoticallyTightBounds} proved that permutation graphs present the following smooth trade-off when instantiated with the permutation that reverses the binary representation of the index $i$ (see Figure~\ref{fig:bit-reversal-1} for an illustration).

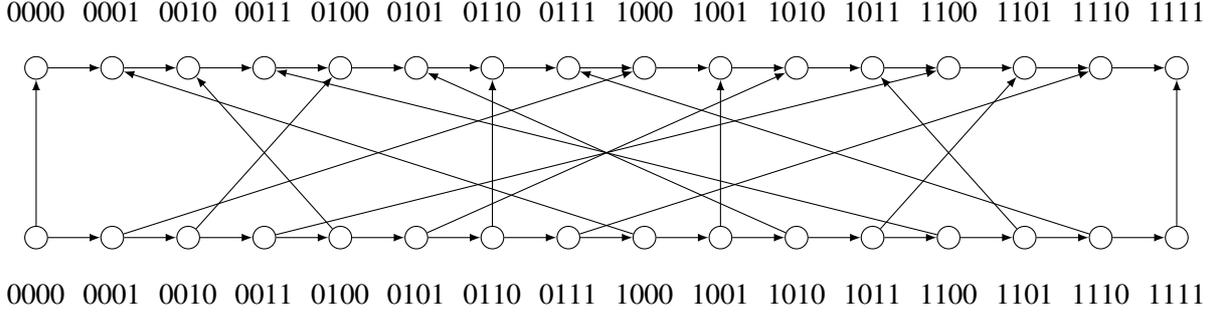
\begin{figure}[tp]
  \centering
  \begin{tikzpicture}[yscale=1.5]
    \foreach \i in {0,1,...,15}{
      \pgfmathsetbasenumberlength{4}
      \pgfmathdectobase\ilabel{\i}{2}
      \node (\i0) [draw,shape=circle,inner sep=3pt] at (\i,0) {};
      \node at (\i,-0.5) {\ilabel};
      \node (\i1) [draw,shape=circle,inner sep=3pt] at (\i,1.5) {};
      \node at (\i,2) {\ilabel};
    }
    \foreach \i in {0,1,...,15}{
      \pgfmathsetbasenumberlength{4}
      \pgfmathdectobase\ilabel{\i}{2}
      \pgfmathsetbasenumberlength{1}
      \reverseto\ilabel\irev
      \pgfmathbasetodec\j{\irev}{2}
      \draw [-latex] (\i0) to (\j1);
}
    \xdef\j{0}
    \foreach \i in {1,2,...,15}{
      \draw [-latex] (\j0.east) to (\i0.west);
      \draw [-latex] (\j1.east) to (\i1.west);
      \pgfmathsetmacro{\j}{\i}
      \xdef\j{\j}
    }
  \end{tikzpicture}
  \caption{A bit-reversal permutation graph.}
  \label{fig:bit-reversal-1}
\end{figure}

\begin{theorem}[\cite{LT82AsymptoticallyTightBounds}]\label{th:LT82bit-reversal}
Let $G$ be a bit-reversal permutation graph on $2n$ vertices.
For any $3\leq \spacestd \leq n$, there is a standard pebbling in space $\spacestd$ and time $\Bigoh{n^2 / \spacestd }$.
Moreover, any standard pebbling $\pebbling$ in space $\spacestd$ satisfies $\pebtime{\pebbling} = \Bigomega{n^2 / \spacestd }$.
\end{theorem}

We show that these graphs also present a smooth reversible pebbling trade-off and, in particular, for $\spacestd = n^{\bigtheta{1}}$ and $\spacestd \leq n$, any reversible pebbling
$\pebbling$ in space $\spacestd$ satisfies $\pebtime{\pebbling_n} = \Bigomega{n^2 / \spacestd}$ and there are matching upper bounds.
To this end, we use the following proposition.

\begin{proposition}\label{pr:pebbling-the-line-in-linear-time}
For every natural number $k$, the line graph over $n$ vertices can be reversibly pebbled using ${2k \cdot n^{1/k}}$ pebbles in time ${2^k \cdot n}$.
\end{proposition}

\begin{proof}
Observe that the line graph over $n$ can be pebbled (in the standard way) using $2$ pebbles in time ${2n}$.
The proposition follows now by applying Theorem~\ref{th:rev-simulation} with $\epsilon = k / \log(n)$. 
\end{proof}

Using Proposition~\ref{pr:pebbling-the-line-in-linear-time}, we obtain the following result.

\begin{theorem}\label{th:constant-space-trade-offs}
  Let $G_n$ be a bit-reversal permutation graph on $2n$ vertices.
  Then $G_n$ satisfies the following properties:
\begin{enumerate}
\item The reversible pebbling price of $G_n$ is at most
    $2\log n + 2$. \label{item:1-permutation}
  \item If $s$    
    satisfies $4\log n \leq s \leq 2n$ and $k$ is such
    that $s= 4k n^{1/k}$, then there is a reversible strategy
    in simultaneous space $s$ and 
    time $\Bigoh{k2^{2k}\cdot n^2/s}$.    \label{item:2-permutation} 
  \item 
    Any  standard pebbling $\pebbling_n$ of $G_n$ 
    satisfies $\pebtime{\pebbling_n} = \Bigomega{n^2/ \pebspace{\pebbling_n}}$.    \label{item:3-permutation}
  \end{enumerate}
  
\end{theorem}

\begin{proof}
The upper bounds (item~\ref{item:1-permutation} and item~\ref{item:2-permutation}) hold for any permutation graph.

For item~\ref{item:1-permutation}, we can simulate a reversible pebbling of the top line that uses space at most $\log n + 1$ (as per Proposition~\ref{pr:rev-peb-price-path}), and every time we need a pebble on a vertex $v$ of the bottom line in order to place or remove a pebble on the top line, we reversibly pebble the bottom line until $v$ is pebbled (can be done with $\log n + 1$ pebbles), make the move on the top line, and then unpebble the bottom line.

To obtain item~\ref{item:2-permutation}, we consider a two stage strategy. In the first phase, we place $n^{1/k}$ pebbles spaced equally apart on the bottom line. We refer to these pebbles as \emph{fixed} pebbles, since they will remain on the graph until the sink is pebbled. In the second phase, we simulate a reversible pebbling on the top line with $2kn^{1/k}$ pebbles and every time we need a pebble on a vertex $v$ on the bottom line to make a move on the top line, we reversibly pebble $v$ (with $2(k-1)n^{1/k}$ pebbles) from the nearest fixed pebble, make the move on the top line, and then unpebble the segment on the bottom line.

All that is left to show is that this can be done within the space budget of $4kn^{1/k}$ in time $\bigoh{2^{2k}\cdot n^2/s}$. For the first phase, we reversibly pebble $n^{1/k}$ segments of length $m = n^{1-1/k}$. By Proposition~\ref{pr:pebbling-the-line-in-linear-time}, each of the segments can be reversibly pebbled using $2(k-1)n^{1/k} = 2(k-1)m^{k-1}$ pebbles in time $2^{k-1} n^{1-1/k}$. Since every segment must be pebbled and then unpebbled, the total time for the first phase is $2 \cdot 2^{k-1} n^{1-1/k} \cdot n^{1/k} = 2^k n$, and the total number of pebbles used is less than $2kn^{1/k}$: $n^{1/k}$ for the fixed pebbles and $2(k-1)n^{1/k}$ for pebbling each segment.

We turn to analyze the second phase. By Proposition~\ref{pr:pebbling-the-line-in-linear-time}, the top line can be reversibly pebbled in simultaneous space $2kn^{1/k}$ and time $2^{k} n$. For each move in the top line, we need to pebble and unpebble a segment of length at most $n^{1-1/k}$. As argued before, this can be done in simultaneous space $2(k-1)n^{1/k}$ and time $2 \cdot 2^{k-1} n^{1-1/k}$. Therefore, at any point in the pebbling strategy there are at most $2kn^{1/k}$ pebbles on the bottom line and at most $2kn^{1/k}$ pebbles on the top line, and the total time of the pebbling is at most $2^k n+ 2^{2k} n^{2-1/k} \leq {4k2^{2k} n^2/s}$.

Finally, item~\ref{item:3-permutation} follows from the standard pebbling lower bound in Theorem~\ref{th:LT82bit-reversal}. 
\end{proof}

From Theorem~\ref{th:constant-space-trade-offs} we obtain the following corollary that, together with Theorem~\ref{thm:pebbling-nullstellensatz-correspondence}, implies Theorem~\ref{th:NS-permutation}.

\begin{corollary}
Any reversible 
    pebbling strategy $\pebbling_n$ for $G_n$ 
    that optimizes
    time given space constraint 
    $n^{\bigtheta{1}}$ (and less than $n$) 
    exhibits a trade-off
    $
    \pebtime{\pebbling_n} = 
    \Bigtheta{n^2 / \pebspace{\pebbling_n}}
    $. \end{corollary}
  
}

\section{Concluding Remarks}
\label{sec:conclusion}

In this paper we prove that size and degree of Nullstellensatz refutations 
in any field of pebbling
formulas are exactly captured by 
time and space of the reversible pebble
game on the underlying graph.
This allows us to prove a number of strong size-degree trade-offs for
Nullstellensatz. To the best of our understanding no such results have
been known previously.

The most obvious, and also most interesting, open question is whether
there are also size-degree trade-offs for the stronger polynomial
calculus proof system. 
Such trade-offs cannot be exhibited by the pebbling formulas
considered in this work, since such formulas have small-size
low-degree polynomial calculus refutations, but the formulas
exhibiting size-width trade-offs for
resolution~\cite{Thapen16Trade-off} appear to be natural candidates.\footnote{Such a result was very recently reported 
  in~\cite{LNSS20Tradeoffs}.}

Another interesting question is whether the tight relation between
Nullstellensatz and reversible pebbling could make it possible to
prove even sharper trade-offs for size versus degree in
Nullstellensatz, where just a small constant drop in the degree would
lead to an exponential blow-up in size.
Such results for pebbling time versus space are known for the standard
pebble game, e.g., in~\cite{GLT80PebblingProblemComplete}.
It is conceivable that a  similar idea could be applied to the
reversible pebbling reductions in~\cite{CLNV15Hardness}, but it is not
obvious whether just adding a small amount of space makes it possible
to carry out the reversible pebbling time-efficiently enough.

Finally, it can be noted that our results crucially depend on that we
are in a setting with variables only for positive literals.  For
polynomial calculus it is quite common to consider the stronger
setting with ``twin variables'' for negated literals (as in the
generalization of polynomial calculus in~\cite{CEI96Groebner} to
\introduceterm{polynomial calculus resolution}
in~\cite{ABRW02SpaceComplexity}). 
It would be nice to generalize our size-degree trade-offs for
Nullstellensatz to this setting,
but it seems that some additional ideas are needed to make this work.
 
\section*{Acknowledgements}

We are grateful for many interesting discussions about matters 
pebbling-related (and not-so-pebbling-related)
with
Arkadev Chattopadhyay,
Toniann Pitassi,
and
Marc Vinyals.

This work was mostly carried out while the authors were visiting the
Simons Institute for the Theory of Computing in association with the
DIMACS/Simons Collaboration on Lower Bounds in Computational
Complexity, which is conducted with support from the National Science
Foundation.

Or Meir was supported by the Israel Science Foundation (grant No. 1445/16).
Robert Robere was  supported by NSERC, and also conducted part of this work at
DIMACS with support from the National Science Foundation under grant number CCF-1445755.
Susanna F.~de Rezende and Jakob Nordström were supported by the 
\emph{Knut and Alice Wallenberg} grant
KAW 2016.0066
\emph{Approximation and Proof Complexity}.
Jakob Nordström was also supported by the 
Swedish \mbox{Research} Council grants
\mbox{621-2012-5645}
and
\mbox{2016-00782}.

 \newcommand{\etalchar}[1]{$^{#1}$}

\end{document}